\documentclass{hustmedia-paper}

\setcopyright{none}
\usepackage{booktabs}
\usepackage{amsmath}
\usepackage{mathtools}
\usepackage{amsthm}
\usepackage{enumitem}
\usepackage{multirow}
\usepackage{tabularx}
\usepackage{colortbl}

\newcommand{\dcmark}{\textcolor{green!50!black}{\ensuremath{\checkmark}}}
\newcommand{\dxmark}{\textcolor{red!75!black}{\ensuremath{\times}}}

\usepackage{algorithm}
\usepackage{algpseudocode}
\usepackage[table]{xcolor}
\definecolor{basepurple}{HTML}{bdb2ff}
\usepackage{xcolor}
\definecolor{CMTgreen}{RGB}{0,120,0}
\definecolor{lightgreen}{RGB}{220,255,220}
\definecolor{lightred}{RGB}{255,220,220}
\usepackage[normalem]{ulem}

\usepackage[capitalize,noabbrev]{cleveref}

\theoremstyle{plain}
\newtheorem{theorem}{Theorem}[section]
\newtheorem{proposition}[theorem]{Proposition}

\theoremstyle{definition}

\theoremstyle{remark}

\HMMAuthorLine{
Yunyao Zhang\textsuperscript{1,*} \textbullet\
Yihao Ai\textsuperscript{2,*} \textbullet\
Boqi Xu\textsuperscript{1} \textbullet\
Zuocheng Ying\textsuperscript{1} \textbullet\
Junqing Yu\textsuperscript{1} \textbullet\
Zikai Song\textsuperscript{1,\ensuremath{\dagger}}
}

\HMMAffiliationLine{
\textsuperscript{1}Huazhong University of Science and Technology;
\textsuperscript{2}National University of Singapore
}

\HMMAuthorNote{
\textsuperscript{*} Equal contribution.
\qquad
\textsuperscript{\ensuremath{\dagger}} Corresponding author:
\HMMEmail{skyesong@hust.edu.cn}.
}
\HMMPublicationStatus{Preprint}

\HMMCorrespondingEmail{skyesong@hust.edu.cn}

\HMMAbstractEmailOn

\HMMAuthorInformation{
Yunyao Zhang, Zuocheng Ying, Junqing Yu, Wei Yang, and Zikai Song ---
Huazhong University of Science and Technology, Wuhan, China;
Yihao Ai --- National University of Singapore, Singapore;
Qirui Mi --- Institute of Automation, Chinese Academy of Sciences, Beijing, China.
}

\HMMDisplayTitle{Rethinking Social Simulation in Open-Population Systems: From Persistent Agents to Communication Processes}

\HMMProjectWebsite{https://ai4ss.github.io/OpenSim/}

\title[OpenSim]{Rethinking Social Simulation in Open-Population Systems: From Persistent Agents to Communication Processes}

\author{Author One}
\affiliation{
  \institution{Huazhong University of Science and Technology}
  \city{Wuhan}
  \country{China}}
\email{author@example.com}

\author{Author Two}
\affiliation{
  \institution{Huazhong University of Science and Technology}
  \city{Wuhan}
  \country{China}}

\begin{abstract}
Large language models (LLMs) have enabled social simulation with persistent agents, where fixed computational individuals maintain internal states and interact over time.
However, large-scale social networks is inherently \textbf{open-population}: the participating population changes over time, users enter and leave dynamically, and only a sparse subset becomes observable through communication at any given time.
In such systems, persistent account identities do not provide continuously observable individual-state trajectories, making recursively maintained psychological states insufficiently grounded in empirical observations.
We therefore shift the simulation target from persistent individual trajectories to the observable evolution of social communication, modeling who participates, what they communicate, and how these communications shape subsequent system dynamics.
We instantiate this formulation with participant activation from a global user pool, collective state transition, and an LLM policy for communication generation, while explicitly separating environmental inputs from endogenous system evolution.
Across five Qwen backbones, our framework reduces average KL divergence by \textbf{71.2\%} over Direct LLM and \textbf{35.0\%} over MF-LLM, with boundary and counterfactual experiments further validating the proposed modeling boundaries.
Our results support communication-centered simulation as a
principled formulation for long-horizon open-population social systems.
\end{abstract}

\begin{document}
\maketitle

\section{Introduction}
Large-scale online platforms such as Twitter, Weibo, and live-streaming communities have become major environments for public opinion and collective response~\citep{Social-Simulation-overview-2014,social-network-analysis2004development}.
Large language models (LLMs) increasingly enable simulation of these systems through agents that generate socially situated actions and content~\citep{LLM-Agent-based-simulation-survey-2024large,zhang2026too,SIM2026}.
A common formulation represents people as persistent agents and carries the social process forward through their individual histories or latent states~\citep{Stanford-town-2023,Stanford1000agents-2024,Agentsociety-2025agentsociety}.
We call this a \emph{human-trajectory-centered} formulation.
Its empirical interpretation, however, remains unclear~\citep{CUP2023-intro1,CUP2023-intro2}.
\textit{If an LLM agent is intended to stand for a person, the validity of its trajectory appears to depend on an unverified equivalence between the language model and the human participant}.
If it is not, the empirical referent of the generated individual trajectory becomes difficult to specify.
The central question is therefore whether persistent individual trajectories are the appropriate recursive object for an open social process.

This question is particularly pressing on open online platforms.
For any given social event, the potential participant population is massive, while actual participation is sparse and time-varying~\citep{user-behavior-2013understanding,5rule-lupeng-2018exploring}.
Users may enter at any time and contribute only briefly before disengaging, even while remaining active elsewhere on the platform~\citep{szell2014contraction}.
Individuals are therefore observed intermittently within an event, whereas the platform process continues through posts, replies, and shares~\citep{twitter-2010,NC-Social-Media-2022universality}.
These \emph{communicative events} accumulate into a system-visible history that conditions subsequent activity.
Although a human-trajectory-centered formulation can be coherent for relatively closed populations with sustained participation, its interpretation becomes less clear in this open setting.
It leaves three basic scientific questions unresolved: \emph{what the simulator predicts}, \emph{what it is responsible for}, and \emph{what the LLM represents}~\citep{ecosystems-2005-Science,simulation-model-2003}.

\textbf{What the simulator predicts.}
A human-trajectory-centered formulation requires individual trajectories to be continued over time, yet in open online systems individuals are only observed when they act. The unobserved evolution between actions therefore has neither direct supervision nor an observable target for validation. This leaves the simulator with a dilemma: either generate latent individual evolution that cannot be empirically verified, or obtain future participants and behavior from the realized trajectory, thereby revealing part of the future. 

\textbf{What the simulator is responsible for.}
A human-trajectory-centered formulation treats future social change as the continued evolution of its modeled individuals. In an open system, however, future trajectories are shaped by contingencies that have not yet entered the social process. Absorbing these contingencies into individual trajectories collapses possible futures into the single realized continuation the simulator is expected to reproduce. It therefore becomes unclear whether the model is simulating possible social evolution or reconstructing the future that happened to occur.

\textbf{What the LLM represents.}
A human-trajectory-centered formulation requires some model to carry each individual's trajectory forward, and LLM-based simulators assign this role to the LLM. Yet the simulator typically observes only sparse digital traces of a person, such as profile attributes and past activity. These observations may support conditional generation of observable communicative acts, but they do not establish that the LLM recovers the person's psychological evolution. The resulting trajectory may serve as a computational state, but the available observations do not justify interpreting it as the recovered psychological evolution of the represented person.

These are not independent limitations. They arise from the same modeling decision: asking an unobserved individual trajectory to carry the continuity of an observable event history. We therefore shift the recursive object of simulation. \textit{Persons remain the participants through whom communicative events are realized, but their latent psychological trajectories are not treated as persistent system state}. Instead, given a system-visible history and declared environmental inputs, the simulator recursively updates that history and estimates a conditional distribution over subsequent observable record updates. Its inputs are observed or declared, its outputs are empirically addressable, and it is responsible for the generated platform process conditional on environmental inputs, not for inventing external events. The empirical claim therefore no longer depends on establishing equivalence between an LLM and a human mind. We arrive at this formulation from the computational requirements of open-population simulation. Conceptually, however, its central distinction echoes Luhmann's social systems theory: social processes reproduce through communication rather than through the continuous evolution of individual consciousness~\citep{luhmann-1995-socialsystem,esposito2022artificial}.

Our architecture realizes this contract through three components. We model \textbf{endogenous participation} over an open population, learn an \textbf{autonomous communication policy} that generates actions and content from information available at simulation time, and introduce \textbf{future-constrained communication learning} to connect locally plausible acts with subsequent process evolution. Future dynamics are used only as training supervision; at inference, the simulator generates and selects candidates without access to the realized future trajectory. These components recursively generate the open platform process rather than reconstructing presumed human trajectories.

Our contributions are threefold.
\begin{itemize}[leftmargin=1.4em, itemsep=0.25em, topsep=0.3em]
\item We identify a structural mismatch between human-trajectory-centered simulation and open online social systems: ambiguity about prediction, endogenous responsibility, and the interpretation of LLM agents arises from using unobserved individual trajectories to carry a process observed through intermittent platform events.

\item We reformulate open online social simulation as recursive prediction and update of a system-visible event history, with declared conditioning information, event-level prediction targets, and an operational system-environment boundary, without requiring equivalence between an LLM and a human mind, in structural alignment with Luhmann's \textbf{social systems theory}.

\item We instantiate and evaluate this formulation with an open-population simulator combining endogenous participant realization, autonomous action-and-content generation, and future-constrained learning. Across five Qwen backbones, the resulting framework improves long-horizon trajectory fidelity without access to the realized future participant sequence.
\end{itemize}

\begin{figure*}
    \centering
    \includegraphics[width=0.9\linewidth]{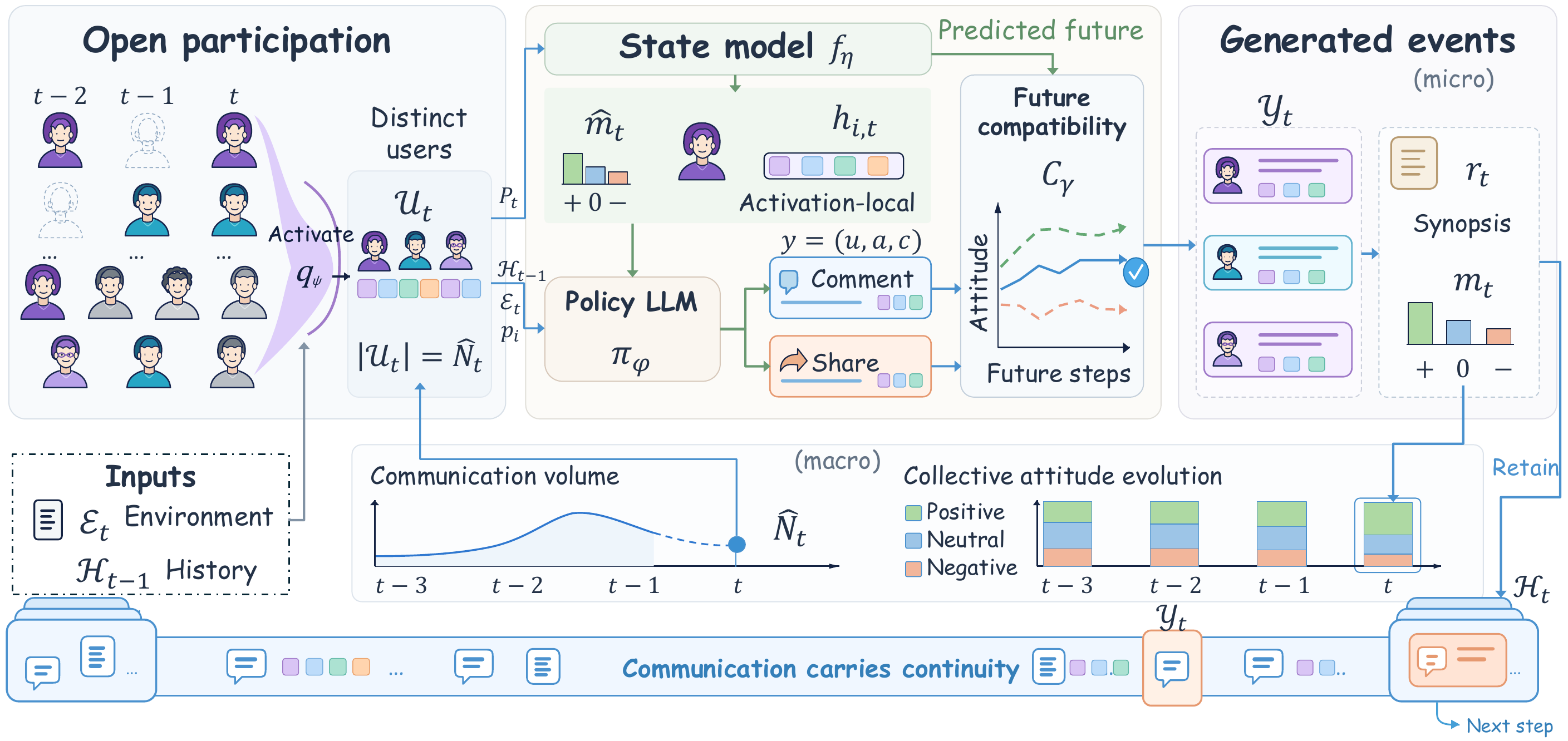}
    \caption{\textbf{Overview of OpenSim.}
OpenSim recursively models participation, mean-field state transition, and future-constrained communication generation, with realized communication updating the collective state and platform history for subsequent steps.}
    \label{fig:placeholder}
\end{figure*}

\section{Background and Problem Formulation}
\label{sec:background}

We formalize open-population social simulation with sparse, time-varying participation (\S\ref{sec:problem}), examine limitations of persistent-agent dynamics in LLM-based simulators (\S\ref{sec:persistent}), and define a simulation target for observable communication with an explicit system--environment boundary (\S\ref{sec:target}).

\subsection{Open-Population Social Systems}
\label{sec:problem}

We consider an online social process evolving over discrete time steps $t\in\{1,\ldots,T\}$.
Let $\mathcal{U}$ denote a global pool of potential participants.Unlike closed-agent simulations, the users participating in an online event change over time.
At each step, only a sparse and time-varying subset $\mathcal{U}_t\subseteq\mathcal{U}$ becomes active, with $N_t=|\mathcal{U}_t|$~\citep{peak-height-lupeng-2018predicting,peak-time-lupeng2018big}.
The same user may reappear across non-consecutive steps.
Let $\mathcal{H}_{t-1}$ denote the observable platform history before step $t$, and $\mathcal{E}_t$ the exogenous information introduced at that step.
For each active user $u_i\in\mathcal{U}_t$, a communicative event is represented as $y_{i,t}=(u_i,a_{i,t},c_{i,t})$, where $a_{i,t}$ is the platform action and $c_{i,t}$ is the expressed or propagated content.
The realized event set is $\mathcal{Y}_t=\{y_{i,t}:u_i\in\mathcal{U}_t\}$, which is incorporated into the observable platform history $\mathcal{H}_t$.

We call this an \textbf{open-population social system}: participation is sparse and intermittent, with a changing active population.
Persistent account identifiers do not imply continuously observable individual-state trajectories; continuity instead resides in the observable platform history $\mathcal{H}_t$.

\subsection{Persistent-Agent Formulation}
\label{sec:persistent}

Existing LLM-based social simulators commonly assume a \emph{closed} agent population, where a fixed set of computational agents is instantiated before the rollout and persistently maintained throughout the simulation~\citep{Socioverse2025socioverse,Oasis-2025}.
Each agent $i$ maintains an internal state $z_{i,t}$, which may encode opinion, emotion, memory, or other latent attributes, and evolves recursively as
\begin{equation}
z_{i,t+1}
\sim
G_{\theta}(z\mid z_{i,t},p_i,\mathcal{H}_t),
\qquad
y_{i,t+1}
\sim
\pi_{\phi}(y\mid z_{i,t+1},p_i,\mathcal{H}_t),
\label{eq:persistent}
\end{equation}
where $p_i$ denotes system-visible participant information.
Such information may condition observable communication, but does not by itself identify a private psychological trajectory.
The resulting process is carried forward as
$
\{z_{i,t}\}
\rightarrow
\{z_{i,t+1}\}
\rightarrow
\mathcal{Y}_{t+1}
\rightarrow
\mathcal{H}_{t+1}.
$

This formulation differs fundamentally from the open-population setting, where only a time-varying subset is active at each step.
For a user absent between two active steps, the system provides no direct observation of the intermediate individual state, so recursively carrying forward $z_{i,t}$ propagates simulator-generated latent states beyond what is identified by the observable platform history.
Fixing the same agents throughout the rollout reduces the problem to a closed-population approximation, whereas supplying future active users from the realized trajectory introduces future participation information~\citep{MF-LLM2025mf,zhang2026coupling}.
Thus, in an open-population system, \textbf{both \emph{who participates} and \emph{what they communicate} must be generated as part of the social process}.

\subsection{Simulation Target and Boundary}
\label{sec:target}

We therefore shift the recursive object from persistent individual trajectories to the observable evolution of communication.
Given the platform history $\mathcal{H}_t$ and declared environmental input $\mathcal{E}_{t+1}$, the simulator generates
\begin{equation}
\mathcal{Y}_{t+1}
\sim
p_{\theta}\!\left(
\cdot\mid\mathcal{H}_t,\mathcal{E}_{t+1}
\right),
\qquad
\mathcal{H}_{t+1}
=
F\!\left(
\mathcal{H}_t,\mathcal{Y}_{t+1},\mathcal{E}_{t+1}
\right).
\label{eq:communication_process}
\end{equation}
Accordingly, the recursive process becomes
$\mathcal{H}_t\rightarrow\mathcal{Y}_{t+1}\rightarrow\mathcal{H}_{t+1}$,
rather than a recursion over persistent latent individual states.

This formulation separates latent individual state $z_{i,t}$, participant activation $\mathcal{U}_t$, and environmental input $\mathcal{E}_t$.
Latent individual states are not treated as persistent recursive system state; participant activation is generated as part of the simulated process; and environmental inputs may affect the process without being endogenous prediction targets.
Operationally, variables generated by the simulator and recursively incorporated into $\mathcal{H}$ belong to the simulated system, whereas externally supplied variables whose generation is not modeled belong to $\mathcal{E}$.

Although motivated by the computational requirements of open-population simulation, this formulation structurally corresponds to Luhmann's view that social systems reproduce through communication rather than through the continuous evolution of individual consciousness~\citep{luhmann-1995-socialsystem}.
From this perspective, the continuity of a social process can be located in recursively connected communications, rather than in continuously maintained internal states of persistent individuals.
Esposito's notion of \emph{artificial communication} further suggests that algorithmic systems may participate in producing communication without being treated as equivalents of human consciousness~\citep{esposito2022artificial}.
This distinction is particularly relevant to LLM-based simulation: an LLM may model the production of observable communication without requiring the stronger assumption that it constitutes a persistent psychological surrogate for a real participant.

\section{Methodology}
\label{sec:method}
We instantiate the proposed formulation with three components: participation modeling (\S\ref{sec:interaction}), collective state transition (\S\ref{sec:state_transition}), and communicative policy modeling (\S\ref{sec:policy}). These components jointly govern participation, collective state, and communication at each simulation step, as summarized in Algorithm~\ref{alg:simulation} in the appendix.

\subsection{Participation Modeling}
\label{sec:interaction}

The first component models the scale and composition of communication at time $t$ by forecasting the communication volume and activating the participants who realize these events.

\noindent \textbf{Communication Volume Forecasting}
Let $N_t$ denote the number of communicative events at time $t$. We forecast the communication volume using information available before the current step:
\begin{equation}
\hat{N}_t
=
f_{\omega}
\left(
\mathcal{H}_{t-1},
\mathcal{E}_t
\right),
\label{eq:interaction_forecast}
\end{equation}
where $\mathcal{H}_{t-1}$ is the accumulated platform history and $\mathcal{E}_t$ is the declared environmental input. The forecasting model summarizes recent communication dynamics and event-level context without accessing the realized future participation volume.

\noindent \textbf{Global-Pool Participant Activation}
Given $\hat{N}_t$, participants are activated from a predefined global user pool $\mathcal{U}=\{u_1,u_2,\ldots,u_M\}$. Each candidate $u_i$ is represented by system-visible information $p_i$, such as profile attributes, interest labels, and historical activity. We assign each candidate an activation likelihood
\begin{equation}
q_{\psi}
\left(
u_i
\mid
p_i,
\mathcal{H}_{t-1},
\mathcal{E}_t
\right),
\qquad
u_i\in\mathcal{U},
\label{eq:participant_activation}
\end{equation}
and obtain $\mathcal{U}_t$ by sampling $\hat{N}_t$ participants without replacement according to these likelihoods, such that $|\mathcal{U}_t|=\hat{N}_t$. Each activated participant realizes one communicative event at the current step. Activation is step-specific, allowing users to appear once, reappear after inactivity, or never reappear; process-level continuity is carried by $\mathcal{H}_t$ rather than by persistently maintained agents.

\subsection{Mean-Field State Transition}
\label{sec:state_transition}

Following M$^2$Sim~\citep{zhang2026coupling}, we represent the collective communication state with a distributional mean field $m_{t-1}\in\Delta^2$ and a textual synopsis $r_{t-1}$ retained in $\mathcal{H}_{t-1}$.
Unlike persistent-agent formulations, participant-level latent representations are inferred only for users activated at the current step.

Given the current participant features $P_t$, platform history $\mathcal{H}_{t-1}$, and environmental input $\mathcal{E}_t$, the transition model predicts
\begin{equation}
\left(
\hat{m}_t,
\{h_{i,t}\}_{u_i\in\mathcal{U}_t}
\right)
=
f_{\eta}
\left(
\mathcal{H}_{t-1},
\mathcal{E}_t,
P_t
\right),
\label{eq:state_transition}
\end{equation}
where $\hat{m}_t$ is the predicted attitude composition of the current communication and $h_{i,t}$ is an activation-local latent representation used by the communication policy.
These representations are not recursively maintained across inactive periods.

After the communicative events $\mathcal{Y}_t$ are generated, their expressed attitudes are aggregated to update the realized mean field $m_t$, while the textual synopsis is updated as
\begin{equation}
r_t
=
g_{\xi}
\left(
r_{t-1},
m_t,
\mathcal{Y}_t,
\mathcal{E}_t
\right).
\label{eq:textual_mean_field}
\end{equation}
Thus, $\hat{m}_t$ guides current communication, whereas $(m_t,r_t)$ records the realized simulated process and is incorporated into $\mathcal{H}_t$ for subsequent steps.

\subsection{Future-Constrained Communication Policy}
\label{sec:policy}

Given an activated participant $u_{t,j}$, the policy generates an action and its associated content:
\begin{equation}
(a_{t,j},c_{t,j})
\sim
\pi_{\phi}\!\left(
\cdot\mid
\mathcal{H}_{t-1},
\mathcal{E}_t,
p_{u_{t,j}},
h_{t,j},
\hat{m}_t
\right),
\label{eq:policy_generation}
\end{equation}
yielding the communicative event
$y_{t,j}=(u_{t,j},a_{t,j},c_{t,j})$.

Local likelihood alone does not ensure that a plausible response is compatible with subsequent collective evolution.
Building on the long-horizon consistency objective of M$^2$Sim~\citep{zhang2026coupling}, we introduce a response-sensitive trajectory critic
\begin{equation}
\widehat{\mathbf{m}}_{t+1:t+K}
=
C_{\gamma}
\left(
\mathcal{H}_{t-1},
\hat{m}_t,
y_{t,j}
\right),
\label{eq:trajectory_critic}
\end{equation}
which predicts the subsequent mean-field trajectory conditioned on the candidate communication.
Given the observed training trajectory $\mathbf{m}_{t+1:t+K}$, the trajectory loss is
\begin{equation}
\mathcal{L}_{\mathrm{traj}}
=
\sum_{k=1}^{K}
D\!\left(
m_{t+k},
\widehat{m}_{t+k}
\right),
\label{eq:trajectory_loss}
\end{equation}
where $D(\cdot,\cdot)$ measures distributional discrepancy.

A trajectory critic may otherwise rely on the shared context while ignoring the response itself.
We therefore construct masked and cross-context mismatched responses $\tilde{y}_{t,j}$ and encourage the observed response to better match its associated future trajectory:
\begin{equation}
\mathcal{L}_{\mathrm{sens}}
=
\max\!\left(
0,\,
\delta
+
D\!\left(
\mathbf{m}_{t+1:t+K},
C_{\gamma}
\left(
\mathcal{H}_{t-1},
\hat{m}_t,
y_{t,j}
\right)
\right)
-
D\!\left(
\mathbf{m}_{t+1:t+K},
C_{\gamma}
\left(
\mathcal{H}_{t-1},
\hat{m}_t,
\tilde{y}_{t,j}
\right)
\right)
\right).
\label{eq:sensitivity_loss}
\end{equation}

After communication SFT and critic initialization, the policy adapters and trajectory critic are jointly optimized with
\begin{equation}
\mathcal{L}_{\mathrm{joint}}
=
\mathcal{L}_{\mathrm{comm}}
+
\lambda_{\mathrm{traj}}\mathcal{L}_{\mathrm{traj}}
+
\lambda_{\mathrm{sens}}\mathcal{L}_{\mathrm{sens}}.
\label{eq:policy_objective}
\end{equation}

At inference time, candidate communications are ranked by the critic according to their compatibility with the state model's predicted future mean-field evolution.
The selected events constitute $\mathcal{Y}_t$ and are used to update the realized mean field and textual synopsis.

\begin{algorithm}[b]
\caption{\textbf{Open-Population Communication Simulation}}
\label{alg:simulation}
\small
\begin{algorithmic}[1]
\Require Global user pool $\mathcal{U}$, horizon $T$;
volume forecaster $f_{\omega}$, activation model $q_{\psi}$,
state transition model $f_{\eta}$, communication policy $\pi_{\phi}$,
trajectory critic $C_{\gamma}$, textual updater $g_{\xi}$;
environmental inputs $\{\mathcal{E}_t\}_{t=1}^{T}$
\State \textbf{Initialize:} platform history $\mathcal{H}_0$ and collective state $(m_0,r_0)$

\For{$t=1$ \textbf{to} $T$}

    \State $\hat{N}_t \gets f_{\omega}(\mathcal{H}_{t-1},\mathcal{E}_t)$
    \Comment{\textcolor{CMTgreen}{Participation modeling}}

    \State $\mathcal{U}_t \gets
    \textsc{Activate}(\mathcal{U},\hat{N}_t;
    q_{\psi},\mathcal{H}_{t-1},\mathcal{E}_t)$

    \State $(\hat{m}_t,\{h_{t,j}\})
    \gets
    f_{\eta}(\mathcal{H}_{t-1},\mathcal{E}_t,P_t)$
    \Comment{\textcolor{CMTgreen}{Mean-field transition}}

    \ForAll{$u_{t,j}\in\mathcal{U}_t$}
        \State $y_{t,j}
        \gets
        \textsc{GenerateSelect}
        (\pi_{\phi},C_{\gamma};
        \mathcal{H}_{t-1},
        \mathcal{E}_t,
        p_{u_{t,j}},
        h_{t,j},
        \hat{m}_t)$
        \Comment{\textcolor{CMTgreen}{Future-constrained policy}}
    \EndFor

    \State $\mathcal{Y}_t\gets\{y_{t,j}\}$

    \State $m_t\gets\textsc{MeanFieldUpdate}(m_{t-1},\mathcal{Y}_t)$
    \Comment{\textcolor{CMTgreen}{Realized mean field}}

    \State $r_t\gets g_{\xi}(r_{t-1},m_t,\mathcal{Y}_t,\mathcal{E}_t)$

    \State $\mathcal{H}_t\gets
    \textsc{Update}(\mathcal{H}_{t-1},\mathcal{E}_t,\mathcal{Y}_t,m_t,r_t)$
    \Comment{\textcolor{CMTgreen}{History update}}

\EndFor

\State \textbf{Output:} $\{\mathcal{Y}_t\}_{t=1}^{T}$
\end{algorithmic}
\end{algorithm}

\section{Experiments}
\label{sec:experiment}

We evaluate the proposed framework from four perspectives:
(1) long-horizon simulation against representative LLM-based social simulation baselines across multiple backbones (\S\ref{sec:main_results});
(2) environmental communication interventions and their effects on subsequent system evolution (\S\ref{sec:interventions});
(3) boundary analysis through perturbations of participant priors and environmental inputs (\S\ref{sec:boundary_analysis});
and (4) ablations of key components for simulation fidelity and long-horizon consistency (\S\ref{sec:ablation}).

\subsection{Experimental Setup}
\label{sec:exp_setup}

\noindent\textbf{Dataset.}
We use an event-level Chinese social-media corpus with intervention-aware annotations, macro-level interaction volume and attitude trajectories, and fine-grained micro-level user interactions. This structure supports joint modeling of external interventions, collective dynamics, and individual communication behavior. Details are provided in Appendix~\ref{app:dataset}.

\noindent\textbf{Backbone Models.}
We evaluate five backbone settings to examine robustness across model families and scales:
Qwen2-1.5B, Qwen2.5-1.5B, Qwen2.5-3B, Qwen2.5-7B, and Qwen3-1.7B.
All methods within the same backbone setting follow the same evaluation protocol.

\noindent\textbf{Metrics.}
We evaluate both distributional trajectory fidelity and action-level prediction.
For collective dynamics, we report KL divergence (\textbf{KL}), Wasserstein distance (\textbf{Wass.}), and Dynamic Time Warping distance (\textbf{DTW}), where lower values indicate better trajectory agreement.
For action-level prediction, we report Macro-F1, Micro-F1, and negative log-likelihood (\textbf{NLL}).

\begin{table*}[t]
\centering
\footnotesize
\begingroup
\definecolor{basepurple}{HTML}{bdb2ff}
\definecolor{BaselineTint}{HTML}{F0F0F0}
\caption[Main results across backbone models.]{\textbf{Main results across backbone models.}
OpenSim / OpenSim$^\dagger$ use event-specific / external participant pools.
$^*$: observed future participant sequence. \textbf{Bold}/\underline{underline}: best/second-best per model.
Parentheses: improvement (\%) over \textcolor{gray}{gray}-shaded MF-LLM;
\textcolor{green!60!black}{green}/\textcolor{red!70!black}{red} = better/worse. Darker purple = larger gains.}
\label{tab:main_results}
\scriptsize
\newcommand{\OSgain}[1]{\,{\fontsize{6.5}{7.5}\selectfont\textcolor{green!60!black}{(+#1\%)}}}
\newcommand{\OSloss}[1]{\,{\fontsize{6.5}{7.5}\selectfont\textcolor{red!70!black}{(-#1\%)}}}
\setlength{\tabcolsep}{2.5pt}
\renewcommand{\arraystretch}{0.965}
\setlength{\aboverulesep}{1pt}
\setlength{\belowrulesep}{1pt}
\begin{tabularx}{\textwidth}{ll*{6}{>{\raggedleft\arraybackslash}X}}
\toprule
\textbf{Model} & \textbf{Method}
& \mbox{\textbf{KL}\,$\downarrow$} & \mbox{\textbf{Wass.}\,$\downarrow$}
& \mbox{\textbf{DTW}\,$\downarrow$} & \mbox{\textbf{Macro-F1}\,$\uparrow$}
& \mbox{\textbf{Micro-F1}\,$\uparrow$} & \mbox{\textbf{NLL}\,$\downarrow$} \\
\midrule
& Direct LLM$^*$
& 2.4090 & 0.3336 & 0.3055 & 0.3441 & 0.5686 & 4.6838 \\
& Social Retrieval$^*$
& 1.6678 & 0.2568 & 0.2063 & 0.3636 & 0.5526 & 5.1355 \\
\rowcolor{BaselineTint}
\cellcolor{white} & MF-LLM$^*$
& \underline{1.2494} & \underline{0.2246} & \underline{0.1773}
& \underline{0.4311} & \underline{0.6295} & 4.6060 \\
\cellcolor{white} & \textbf{OpenSim}$^\dagger$
& 2.0773\OSloss{66.3} & 0.2905\OSloss{29.3} & 0.2256\OSloss{27.2} & 0.4209\OSloss{2.4} & 0.6199\OSloss{1.5} & \cellcolor{basepurple!20}\textbf{4.1414}\OSgain{10.1} \\
\cellcolor{white}\multirow{-5}{*}{\shortstack{Qwen\\2-1.5B}}
& \textbf{OpenSim}
& \cellcolor{basepurple!18}\textbf{1.1352}\OSgain{9.1} & \cellcolor{basepurple!28}\textbf{0.1922}\OSgain{14.4} & \cellcolor{basepurple!58}\textbf{0.1261}\OSgain{28.9}
& \cellcolor{basepurple!6}\textbf{0.4444}\OSgain{3.1} & \cellcolor{basepurple!10}\textbf{0.6616}\OSgain{5.1} & \cellcolor{basepurple!12}\underline{4.3377}\OSgain{5.8} \\
\midrule
& Direct LLM$^*$
& 3.3485 & 0.2844 & 0.3384 & 0.3730 & 0.4651 & 6.3707 \\
& Social Retrieval$^*$
& 2.4112 & 0.2243 & 0.2695 & 0.4569 & 0.5566 & 6.2568 \\
\rowcolor{BaselineTint}
\cellcolor{white} & MF-LLM$^*$
& \underline{1.4739} & \underline{0.1941} & \underline{0.1396}
& \underline{0.5408} & 0.5881 & 5.1429 \\
\cellcolor{white} & \textbf{OpenSim}$^\dagger$
& 1.5739\OSloss{6.8} & 0.2142\OSloss{10.4} & 0.1934\OSloss{38.5} & 0.5001\OSloss{7.5} & \cellcolor{basepurple!4}\underline{0.6008}\OSgain{2.2} & \cellcolor{basepurple!28}\textbf{4.4138}\OSgain{14.2} \\
\cellcolor{white}\multirow{-5}{*}{\shortstack{Qwen\\2.5-1.5B}}
& \textbf{OpenSim}
& \cellcolor{basepurple!56}\textbf{1.0667}\OSgain{27.6} & \cellcolor{basepurple!22}\textbf{0.1724}\OSgain{11.2} & \cellcolor{basepurple!20}\textbf{0.1252}\OSgain{10.3}
& \cellcolor{basepurple!26}\textbf{0.6097}\OSgain{12.7} & \cellcolor{basepurple!12}\textbf{0.6252}\OSgain{6.3} & \cellcolor{basepurple!22}\underline{4.5980}\OSgain{10.6} \\
\midrule
& Direct LLM$^*$
& 3.8755 & 0.3603 & 0.3771 & 0.3138 & 0.4561 & 13.9996 \\
& Social Retrieval$^*$
& 1.8466 & 0.2283 & 0.2591 & 0.3817 & 0.5129 & 14.0303 \\
\rowcolor{BaselineTint}
\cellcolor{white} & MF-LLM$^*$
& 1.5061 & 0.1546 & 0.2311 & 0.4598 & 0.6222 & 4.5446 \\
\cellcolor{white} & \textbf{OpenSim}$^\dagger$
& \cellcolor{basepurple!0}\underline{1.5030}\OSgain{0.2} & \cellcolor{basepurple!6}\underline{0.1495}\OSgain{3.3} & \cellcolor{basepurple!42}\underline{0.1824}\OSgain{21.1}
& \cellcolor{basepurple!2}\underline{0.4652}\OSgain{1.2} & \cellcolor{basepurple!10}\underline{0.6545}\OSgain{5.2} & \cellcolor{basepurple!2}\textbf{4.4980}\OSgain{1.0} \\
\cellcolor{white}\multirow{-5}{*}{\shortstack{Qwen\\2.5-3B}}
& \textbf{OpenSim}
& \cellcolor{basepurple!66}\textbf{1.0048}\OSgain{33.3} & \cellcolor{basepurple!36}\textbf{0.1273}\OSgain{17.7} & \cellcolor{basepurple!82}\textbf{0.1369}\OSgain{40.8}
& \cellcolor{basepurple!20}\textbf{0.5058}\OSgain{10.0} & \cellcolor{basepurple!34}\textbf{0.7278}\OSgain{17.0} & \cellcolor{basepurple!0}\underline{4.5356}\OSgain{0.2} \\
\midrule
& Direct LLM$^*$
& 3.0914 & 0.3107 & 0.3097 & 0.3471 & 0.5581 & 13.0946 \\
& Social Retrieval$^*$
& 2.5142 & 0.2981 & 0.2882 & 0.3701 & 0.5821 & 12.8823 \\
\rowcolor{BaselineTint}
\cellcolor{white} & MF-LLM$^*$
& 1.6049 & \underline{0.2234} & 0.2564
& 0.4380 & 0.6117 & \textbf{4.6823} \\
\cellcolor{white} & \textbf{OpenSim}$^\dagger$
& \cellcolor{basepurple!22}\underline{1.4294}\OSgain{10.9} & 0.2289\OSloss{2.5} & \cellcolor{basepurple!4}\underline{0.2509}\OSgain{2.1}
& \cellcolor{basepurple!4}\underline{0.4455}\OSgain{1.7} & \cellcolor{basepurple!8}\underline{0.6340}\OSgain{3.6} & \underline{4.9849}\OSloss{6.5} \\
\cellcolor{white}\multirow{-5}{*}{\shortstack{Qwen\\2.5-7B}}
& \textbf{OpenSim}
& \cellcolor{basepurple!94}\textbf{0.8509}\OSgain{47.0} & \cellcolor{basepurple!42}\textbf{0.1759}\OSgain{21.3} & \cellcolor{basepurple!78}\textbf{0.1572}\OSgain{38.7}
& \cellcolor{basepurple!22}\textbf{0.4855}\OSgain{10.8} & \cellcolor{basepurple!20}\textbf{0.6752}\OSgain{10.4} & 5.0921\OSloss{8.8} \\
\midrule
& Direct LLM$^*$
& 5.1480 & 0.4279 & 0.4191 & 0.3314 & 0.5236 & 25.7470 \\
& Social Retrieval$^*$
& 4.7743 & 0.3453 & 0.2210 & 0.3478 & 0.5451 & 25.4819 \\
\rowcolor{BaselineTint}
\cellcolor{white} & MF-LLM$^*$
& 2.0902 & 0.2549 & 0.1715 & \underline{0.5095} & 0.6518 & 6.3621 \\
\cellcolor{white} & \textbf{OpenSim}$^\dagger$
& \cellcolor{basepurple!74}\underline{1.3115}\OSgain{37.3} & \cellcolor{basepurple!14}\underline{0.2377}\OSgain{6.7} & \cellcolor{basepurple!16}\underline{0.1583}\OSgain{7.7}
& 0.4865\OSloss{4.5} & \cellcolor{basepurple!10}\underline{0.6825}\OSgain{4.7} & \cellcolor{basepurple!62}\textbf{4.4081}\OSgain{30.7} \\
\cellcolor{white}\multirow{-5}{*}{\shortstack{Qwen\\3-1.7B}}
& \textbf{OpenSim}
& \cellcolor{basepurple!96}\textbf{1.0929}\OSgain{47.7} & \cellcolor{basepurple!60}\textbf{0.1780}\OSgain{30.2} & \cellcolor{basepurple!72}\textbf{0.1098}\OSgain{36.0}
& \cellcolor{basepurple!14}\textbf{0.5445}\OSgain{6.9} & \cellcolor{basepurple!34}\textbf{0.7610}\OSgain{16.8} & \cellcolor{basepurple!60}\underline{4.4414}\OSgain{30.2} \\
\bottomrule
\end{tabularx}
\endgroup
\end{table*}

\subsection{Do We Need Future Participants?}
\label{sec:main_results}

\noindent\textbf{Baselines.}
We compare OpenSim with three representative social simulation baselines, \textit{all of which use the observed participant sequence from the realized trajectory, leaking future participation information}:
(1) \textbf{Direct LLM$^{*}$}, which generates actions directly from the event context~\citep{s3-2023,GAS32025-ga};
(2) \textbf{Social Retrieval$^{*}$}, which conditions generation on retrieved social traces~\citep{Agentsociety-2025agentsociety,Oasis-2025};
and (3) \textbf{MF-LLM$^{*}$}, which uses a recursively updated textual mean field to guide generation~\citep{MF-LLM2025mf}.
For OpenSim, we evaluate two participant pools:
\textbf{OpenSim} uses an event-specific participant pool constructed from the corresponding event's interaction records, while
\textbf{OpenSim$^\dagger$} uses an external participant pool of 15,992 real-user profiles from SocioVerse~\citep{Socioverse2025socioverse}.
Table~\ref{tab:main_results} reports the main simulation results.

\textbf{(1) OpenSim remains superior without access to the realized future participant sequence.}
OpenSim uses only an event-specific candidate pool and generates participation during simulation, whereas all baselines are supplied with the observed participant sequence from the realized trajectory.
Despite this information disadvantage, OpenSim consistently achieves substantially better trajectory fidelity, reducing average KL by \textbf{71.2\%} over Direct LLM and \textbf{35.0\%} over MF-LLM.
This shows that accurate long-horizon simulation does not require knowing in advance which users will participate at each future step.

\textbf{(2) Exact event-specific participants are not necessary for recovering collective dynamics.}
Even when OpenSim$^\dagger$ uniformly samples participants from an event-agnostic global real-user pool, its performance remains close to MF-LLM, which is given the realized participant sequence.
Their average KL is nearly identical (1.5790 vs.\ 1.5849), while OpenSim$^\dagger$ also achieves comparable trajectory distances and action-level performance.
This suggests that sparse system-visible user information can be sufficient for reproducing collective dynamics without requiring the simulator to recover the exact individuals who appear in the realized event trajectory.

\textbf{(3) These patterns are consistent across backbone models.}
The same conclusions hold across all five Qwen backbones, spanning Qwen2, Qwen2.5, and Qwen3 from 1.5B to 7B parameters.
OpenSim consistently outperforms sequence-conditioned baselines in trajectory fidelity, while OpenSim$^\dagger$ remains competitive under the substantially weaker global-pool setting.
The gains therefore reflect the proposed simulation formulation rather than a particular backbone or model scale.

\begin{table*}[t]
\centering
\footnotesize
\begingroup
\definecolor{BaselineTint}{HTML}{F0F0F0}
\definecolor{RemoveTint}{HTML}{FFF1DF}
\definecolor{OppositeTint}{HTML}{EDE4F7}

\caption[Communication interventions across backbones.]{%
\textbf{Communication interventions across backbones.}
\textsc{Original} retains the original intervention,
\textsc{Removed} removes it, and
\textsc{Opposite} replaces it with an intervention of the opposite stance.
\textbf{Bold}/\underline{underline} indicate the best/second-best agreement with the corresponding reference.}
\label{tab:counterfactual}

\scriptsize
\setlength{\tabcolsep}{3pt}
\renewcommand{\arraystretch}{0.965}
\setlength{\aboverulesep}{1pt}
\setlength{\belowrulesep}{1pt}

\begin{tabularx}{\textwidth}{
@{}ll*{6}{>{\centering\arraybackslash}X}
!{\vrule width 0.5pt}
*{2}{>{\centering\arraybackslash}X}@{}}
\toprule
\textbf{Model} & \textbf{Intervention}
& \textbf{KL} & \textbf{Wass.} & \textbf{DTW}
& \mbox{\textbf{Macro-F1}}
& \mbox{\textbf{Micro-F1}}
& \textbf{NLL}
& \mbox{\textbf{Att. Acc.}\,$\uparrow$}
& \mbox{\textbf{Att. DTW}\,$\downarrow$} \\
\midrule

\rowcolor{BaselineTint}
\cellcolor{white} & \textsc{Original}
& \textbf{1.1352} & \textbf{0.1922} & \textbf{0.1261}
& \textbf{0.4444} & \textbf{0.6616} & \textbf{4.3377}
& \textbf{64.4} & \textbf{0.1779} \\
\rowcolor{RemoveTint}
\cellcolor{white} & \textsc{Removed}
& 3.4784 & 0.4974 & 0.4528
& 0.3502 & 0.5383 & 4.7429
& \underline{49.6} & \underline{0.3235} \\
\rowcolor{OppositeTint}
\cellcolor{white}
\multirow{-3}{*}{\shortstack{Qwen\\2-1.5B}}
& \textsc{Opposite}
& \underline{2.0065} & \underline{0.4126} & \underline{0.3652}
& \underline{0.3539} & \underline{0.5509} & \underline{4.6957}
& 35.1 & 0.4242 \\
\midrule

\rowcolor{BaselineTint}
\cellcolor{white} & \textsc{Original}
& \textbf{1.0667} & \textbf{0.1724} & \textbf{0.1252}
& \textbf{0.6097} & \textbf{0.6252} & \textbf{4.5980}
& \textbf{65.6} & \textbf{0.1717} \\
\rowcolor{RemoveTint}
\cellcolor{white} & \textsc{Removed}
& 2.4076 & \underline{0.3495} & \underline{0.2864}
& \underline{0.3771} & \underline{0.5656} & \underline{5.1793}
& \underline{36.8} & \underline{0.4212} \\
\rowcolor{OppositeTint}
\cellcolor{white}
\multirow{-3}{*}{\shortstack{Qwen\\2.5-1.5B}}
& \textsc{Opposite}
& \underline{2.2847} & 0.3693 & 0.3089
& 0.3741 & 0.5653 & 5.4473
& 35.1 & 0.4671 \\
\midrule

\rowcolor{BaselineTint}
\cellcolor{white} & \textsc{Original}
& \textbf{1.0929} & \textbf{0.1780} & \textbf{0.1098}
& \textbf{0.5445} & \textbf{0.7610} & \textbf{4.4414}
& \textbf{69.1} & \textbf{0.1241} \\
\rowcolor{RemoveTint}
\cellcolor{white} & \textsc{Removed}
& 3.9696 & \underline{0.2611} & \underline{0.2606}
& \underline{0.3804} & \underline{0.6870} & \underline{4.7457}
& \underline{41.7} & \underline{0.4810} \\
\rowcolor{OppositeTint}
\cellcolor{white}
\multirow{-3}{*}{\shortstack{Qwen\\3-1.7B}}
& \textsc{Opposite}
& \underline{2.0284} & 0.3300 & 0.3328
& 0.3570 & 0.6582 & 4.7813
& 36.4 & 0.4942 \\
\bottomrule
\end{tabularx}

\par\vspace{2pt}
\begin{minipage}{\textwidth}
\scriptsize
\textit{Notes.}
Attitude metrics are computed on 57 common supported event-days after warm-up.
Att.\ Acc.\ measures dominant-class agreement, and Att.\ DTW uses path-normalized Euclidean costs.
Both are text-derived attitude proxies rather than native state-transition states.
\end{minipage}

\endgroup
\end{table*}

\begin{figure*}[t]
    \centering
    \includegraphics[width=\textwidth]{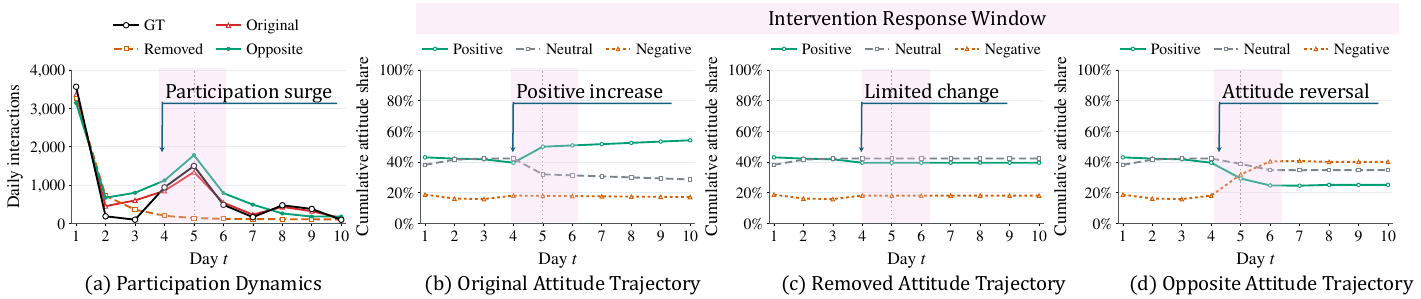}
    \caption{\textbf{Example of communication intervention dynamics.}
    (a) Participation dynamics across intervention settings.
    (b--d) Cumulative attitude trajectories under \textsc{Original}, \textsc{Removed}, and \textsc{Opposite}, respectively.
    The dashed line marks the intervention point.}
    \label{fig:intervention_case}
\end{figure*}


\subsection{How Do Communication Interventions Reshape System Evolution?}
\label{sec:interventions}

\textbf{Setting: Controlled Communication Interventions.}
For each test event, we identify a consequential external communication and construct three conditions while keeping the pre-intervention history unchanged: \textsc{Original} retains the intervention, \textsc{Removed} removes it, and \textsc{Opposite} replaces it with an intervention of the opposite stance.
After the intervention point, participant activation, state transition, and communication generation are autonomously rolled out by the same simulator.
\textsc{Original} uses one rollout per event, while \textsc{Removed} and \textsc{Opposite} average four rollouts.

\textbf{Facet A: Removal primarily attenuates participation and distributional support.}
\textsc{Removed} produces the largest KL divergence across all three backbones, averaging $3.2852$, compared with $2.1065$ for \textsc{Opposite} and $1.0983$ for \textsc{Original}.
Removing a consequential communication suppresses the downstream activation associated with that intervention, reducing participation and distributional support in the subsequent process.
This produces the strongest overall distributional mismatch, while the resulting trajectory changes are comparatively less directional.

\textbf{Facet B: Opposite interventions redirect attitude and temporal evolution.}
\textsc{Opposite} preserves the intervention signal but reverses its stance, causing attitudes, emotions, and expressed positions to evolve in the opposite direction.
On Qwen2.5-1.5B and Qwen3-1.7B, it produces larger Wasserstein and DTW deviations than \textsc{Removed}; it yields the largest attitude discrepancy on every backbone, with the lowest average Att.\ Acc.\ ($35.5\%$) and highest Att.\ DTW ($0.4618$), compared with $42.7\%$ and $0.4086$ for \textsc{Removed}.
Thus, removal primarily attenuates the communication process, whereas an opposite intervention more strongly redirects its semantic and attitudinal trajectory.

\subsection{What Information Should the Simulator Receive?}
\label{sec:boundary_analysis}

\textbf{Setting: Boundary-Controlled Comparison.}
We isolate two modeling boundaries while keeping the state-transition and communication-policy modules fixed:
(1) whether participation is generated by the simulator or supplied from the future participant sequence in the realized trajectory, and
(2) whether consequential external information enters as a declared environmental input or is reconstructed by a simulated information-source agent.
Their Cartesian product yields four controlled settings.
Under this decomposition, M$^2$Sim~\citep{zhang2026coupling} corresponds to the boundary assumptions that combine simulated information-source generation with the observed future participant sequence.
Table~\ref{tab:boundary_analysis} reports the results across three backbones.

\textbf{Fact C: Declared environmental inputs preserve the system--environment boundary.}
When participation is generated, supplying consequential external information as a declared environmental input reduces average KL from $1.3283$ to $1.0983$ ($17.3\%$), Wasserstein distance from $0.1980$ to $0.1809$, and DTW from $0.1339$ to $0.1204$.
The same pattern remains when the future participant sequence is supplied, where average KL decreases from $0.8390$ to $0.7094$.
External contingencies are therefore modeled more faithfully when they enter through the declared environment and the simulator predicts the communication that follows, rather than being reconstructed as endogenous products of the simulated population.

\textbf{Fact D: Future participation information improves reconstruction accuracy.}
Supplying the realized future participant sequence substantially improves trajectory matching under both environmental settings.
With declared environmental inputs, average KL decreases from $1.0983$ to $0.7094$ ($35.4\%$); when the external information is instead generated by a simulated source, it decreases from $1.3283$ to $0.8390$ ($36.8\%$).
The numerical gain therefore comes from information unavailable in genuine forward simulation: knowing in advance who will participate.
This explains why future-conditioned settings can achieve stronger reconstruction accuracy while solving an easier problem than the open-population simulation considered in our main experiments.

\begin{table*}[t]
\centering
\footnotesize
\begingroup
\definecolor{MainTint}{HTML}{F1EDFF}
\definecolor{BaselineTint}{HTML}{F0F0F0}

\caption[Boundary-controlled comparison across backbones.]{%
\textbf{Boundary-controlled comparison across backbones.}
\textsc{Declared} introduces consequential external information as a declared environmental input, whereas \textsc{Simulated} requires the simulated population to generate it.
\textsc{Generated} predicts participation during rollout, whereas \textsc{Future Sequence} supplies the observed participant sequence from the realized trajectory.
$^\dagger$ denotes the participant--environment assumptions corresponding to M$^2$Sim~\citep{zhang2026coupling}.
\textbf{Bold}/\underline{underline} indicate the best/second-best result within each backbone.}
\label{tab:boundary_analysis}

\scriptsize
\setlength{\tabcolsep}{3.2pt}
\renewcommand{\arraystretch}{0.98}
\setlength{\aboverulesep}{1pt}
\setlength{\belowrulesep}{1pt}

\begin{tabularx}{\textwidth}{
@{}lll*{6}{>{\centering\arraybackslash}X}@{}}
\toprule
\textbf{Model} &
\textbf{Environment} &
\textbf{Participation} &
\textbf{KL} $\downarrow$ &
\textbf{Wass.} $\downarrow$ &
\textbf{DTW} $\downarrow$ &
\mbox{\textbf{Macro-F1} $\uparrow$} &
\mbox{\textbf{Micro-F1} $\uparrow$} &
\textbf{NLL} $\downarrow$ \\
\midrule

\rowcolor{MainTint}
\cellcolor{white} &
\textsc{Declared} & \textsc{Generated}
& 1.1352 & 0.1922 & 0.1261
& \underline{0.4444} & \textbf{0.6616} & 4.3377 \\

\cellcolor{white} &
\textsc{Declared} & \textsc{Future Sequence}
& \textbf{0.7686} & \textbf{0.1869} & \textbf{0.1219}
& \textbf{0.4479} & \underline{0.6416} & \textbf{4.0911} \\

\cellcolor{white} &
\textsc{Simulated} & \textsc{Generated}
& 1.2633 & 0.2045 & 0.1392
& 0.4149 & 0.6059 & 4.3448 \\

\rowcolor{BaselineTint}
\cellcolor{white}
\multirow{-4}{*}{\shortstack{Qwen\\2-1.5B}}
& \textsc{Simulated} & \textsc{Future Sequence}$^\dagger$
& \underline{0.8575} & \underline{0.1892} & \underline{0.1240}
& 0.4260 & 0.6363 & \underline{4.1696} \\
\midrule

\rowcolor{MainTint}
\cellcolor{white} &
\textsc{Declared} & \textsc{Generated}
& 1.0667 & \underline{0.1724} & \textbf{0.1252}
& \underline{0.6097} & \underline{0.6252} & 4.5980 \\

\cellcolor{white} &
\textsc{Declared} & \textsc{Future Sequence}
& \textbf{0.7395} & \textbf{0.1674} & \underline{0.1295}
& \textbf{0.6339} & \textbf{0.6471} & \underline{4.3250} \\

\cellcolor{white} &
\textsc{Simulated} & \textsc{Generated}
& 1.3847 & 0.1906 & 0.1376
& 0.5977 & 0.6064 & 4.3886 \\

\rowcolor{BaselineTint}
\cellcolor{white}
\multirow{-4}{*}{\shortstack{Qwen\\2.5-1.5B}}
& \textsc{Simulated} & \textsc{Future Sequence}$^\dagger$
& \underline{0.9327} & 0.1860 & 0.1301
& 0.6080 & 0.6128 & \textbf{4.2038} \\
\midrule

\rowcolor{MainTint}
\cellcolor{white} &
\textsc{Declared} & \textsc{Generated}
& 1.0929 & \underline{0.1780} & \textbf{0.1098}
& 0.5445 & \underline{0.7610} & 4.4414 \\

\cellcolor{white} &
\textsc{Declared} & \textsc{Future Sequence}
& \textbf{0.6202} & \textbf{0.1726} & \underline{0.1168}
& \textbf{0.5772} & \textbf{0.7641} & \textbf{4.2779} \\

\cellcolor{white} &
\textsc{Simulated} & \textsc{Generated}
& 1.3370 & 0.1989 & 0.1249
& 0.5235 & 0.7339 & \underline{4.3599} \\

\rowcolor{BaselineTint}
\cellcolor{white}
\multirow{-4}{*}{\shortstack{Qwen\\3-1.7B}}
& \textsc{Simulated} & \textsc{Future Sequence}$^\dagger$
& \underline{0.7269} & 0.1817 & 0.1240
& \underline{0.5459} & 0.7411 & 4.3711 \\
\bottomrule
\end{tabularx}
\endgroup
\end{table*}

\begin{figure*}[t]
    \centering
    \includegraphics[width=\textwidth]{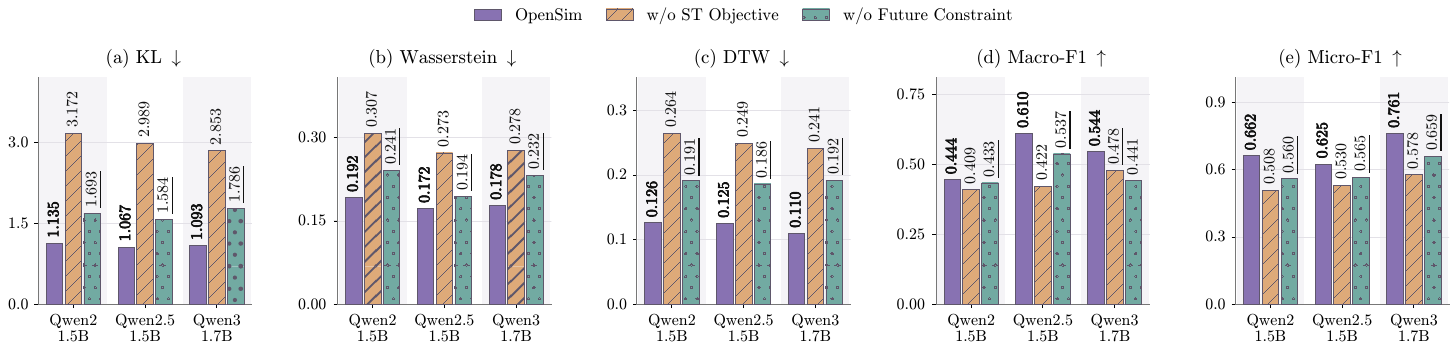}
    \caption{\textbf{Component ablations across backbones.}
    We remove state-transition supervision or the future-constrained policy objective.  
    \textbf{Bold}/\underline{underlined} values mark the best/second-best results per backbone.
    Full results: Appendix Table~\ref{tab:component_ablation}.}
    \label{fig:component_ablation}
\end{figure*}

\subsection{What Drives OpenSim? Component Ablations}
\label{sec:ablation}

\textbf{Setting: Component Ablations.}
We ablate two objectives while keeping participation modeling and the remaining simulation pipeline unchanged:
\textsc{w/o ST Objective} removes the state-transition objective, while
\textsc{w/o Future Constraint} removes the future-constrained policy objective.
Figure~\ref{fig:component_ablation} reports results across three backbones.

\textbf{Fact E: State-transition supervision anchors long-horizon collective dynamics.}
Removing the state-transition objective causes the largest degradation across all three backbones.
Average KL increases from $1.0983$ to $3.0043$ ($+173.5\%$), Wasserstein distance from $0.1809$ to $0.2858$ ($+58.0\%$), and DTW from $0.1204$ to $0.2517$ ($+109.1\%$).
Average Micro-F1 falls from $0.6826$ to $0.5385$, indicating that the degradation extends across distributional and predictive metrics.
These results show that the mean-field state transition provides the primary cross-step structure for maintaining coherent collective evolution, linking successive population states throughout the simulated trajectory.

\textbf{Fact F: Future constraints improve trajectory fidelity beyond local likelihood.}
Removing the future constraint increases average KL to $1.6877$ ($+53.7\%$) and DTW to $0.1899$ ($+57.8\%$), while reducing average Micro-F1 to $0.5943$.
In contrast, its average NLL improves from $4.4590$ to $4.2833$, with particularly lower NLL on Qwen2.5-1.5B and Qwen3-1.7B.
The divergence between local likelihood and trajectory metrics shows that more likely  communications do not necessarily yield more faithful long-horizon social evolution, motivating the future-constrained communication objective.

\section{Conclusion}

We revisit what should serve as the recursive object of simulation in open online social systems. Rather than carrying the process through persistent synthetic individuals and their latent trajectories, we formulate social simulation around the observable evolution of communication, where participation, collective state, and communicative events are generated from information available at simulation time under an explicit system--environment boundary. This formulation preserves individuals as participants without requiring their continuously evolving psychological states to function as persistent system state. Its structure is also consistent with a communication-centered view of social systems, while remaining a computational formulation derived from the open-population setting itself. We hope this perspective provides a clearer foundation for building and evaluating long-horizon social simulators whose predictions remain empirically addressable and whose modeling responsibilities are explicitly defined.

\section*{AI Use Statement}

We used generative AI tools to assist with language editing, translation, and improving the clarity of the manuscript. LLMs were also used as components of the simulation framework and for semantic annotation in evaluation, as described in the methodology and experimental settings. All AI-assisted content was reviewed by the authors, who take full responsibility for the final manuscript.

\section*{Ethics Statement}

This work studies open-population social simulation for analyzing collective communication and intervention responses. Such techniques could be misused for targeted persuasion or manipulation at scale. Simulated communications and states should not be interpreted as verified psychological attributes of real individuals, and LLM-generated outputs may contain biases or errors. We therefore encourage appropriate ethical review, transparency, and privacy safeguards in real-world applications.

\bibliographystyle{ACM-Reference-Format}
\bibliography{main}

\newpage
\appendix
\onecolumn

\begin{center}
{\Large \textbf{Appendix}}
\end{center}

\noindent\textbf{A.} \hyperref[app:A]{Related Work}\par
\noindent\textbf{B.} \hyperref[app:B]{Theoretical Foundations of the Simulation Target}\par
\noindent\textbf{C.} \hyperref[app:C]{Derivations of the Simulation Method}\par
\noindent\textbf{D.} \hyperref[app:D]{Detailed Experimental Setup}\par
\noindent\textbf{E.} \hyperref[app:E]{Additional Experimental Results and Analysis}\par

\section{Related Work}
\label{app:A}

We review four research strands in their broad historical progression, through a common boundary: \textit{which variables are predicted, which are supplied as external conditions, and at what level the simulated state is defined}.
For the task considered here, future participation and communication are prediction targets, external events are declared conditions, and opinion trajectories are observables derived from communication.
The distinctions below concern mismatches with this task, rather than a blanket rejection of earlier simulation objectives.

\subsection*{Diffusion Forecasting: Predicting Activity without Generating Language}
Hawkes processes model event arrivals through history-dependent intensities~\citep{hawkes1971spectra}, with subsequent work distinguishing endogenous amplification from responses to external shocks~\citep{crane2008robust}.
\textit{SEISMIC} predicts eventual popularity from early reshares, while \textit{DeepCas} learns future cascade growth from observed diffusion structure~\citep{zhao2015seismic,li2017deepcas,wang2026seeingwiderjointspatiotemporal}.
These methods explicitly predict future activity and evaluate it against held-out observations.
Neural point processes extend this line to typed, timestamped event generation, and \textit{HawkesTopic} incorporates textual topics into diffusion modeling~\citep{mei2017neural,he2015hawkestopic,deem-2024}. More recently, \textit{PopSim} uses LLM-based multi-agent simulation to predict content popularity from simulated propagation dynamics~\citep{PopSim-2025-liuwu}.
The prediction target nevertheless \textbf{stops short of the communication process}: event times, cascade sizes, and topic representations do not specify what participants will say or how their expressed opinions will evolve.
\textit{This strand forecasts social activity while leaving its language-mediated development ungenerated}.

\subsection*{LLM Social Sandboxes: Persistent Agents without Participation Forecasting}
LLMs~\citep{LPT-2026-ACL,song2025temporal,song2026hypergraph} introduce natural-language memory, planning, and interaction into agent-based simulation, exemplified by \textit{Generative Agents}~\citep{Stanford-town-2023}.
\textit{OASIS} and \textit{AgentSociety} scale this paradigm to larger social environments~\citep{Oasis-2025,Agentsociety-2025agentsociety}.
\textit{Their central modeling unit is a persistent agent: a population is instantiated, and its members repeatedly interact and update internal states}.
\textit{DynamiX} models evolving social ties and agent roles, while \textit{GASim} accelerates large-scale simulation through hybrid LLM and graph-based updates~\citep{sun2025dynamix,zhou2026gasim}.
\textit{CyberJurors} instead simulates multi-agent deliberation and voting for collective dispute resolution~\citep{sun2026cyberjurors}.
Other work examines agents' social cognition, cognitive biases, and behavioral validity~\citep{liu2025exploring,liu2026cobra,liu2026validatedbehavioralhypotheseslens}.
Open online events pose a different structural problem.
Participation is sparse and time-varying: accounts enter, disappear, and return, while future contributors need not belong to the initial population.
\textbf{Who appears at each timestep is a prediction target}, not merely a scheduling choice within an existing cast.
\textit{Generating activity among persistent agents leaves event-specific participation forecasting unresolved, including its coupling to future communications and observable opinion trajectories}.

\subsection*{Mean-Field Simulators: Rendering Content on Supplied Participation}
Mean-field LLM simulators pursue quantitative trajectory alignment by coupling individual decisions with evolving population signals.
\textit{MF-LLM} learns a textual mean field, while \textit{M$^2$Sim} introduces latent micro states and long-horizon consistency training~\citep{MF-LLM2025mf,zhang2026coupling}.
Their response-generation loop operates on the active set supplied at each step; M$^2$Sim explicitly takes active-pool sizes and descriptors as inputs.
Under realized-participation conditioning, this interface \textbf{places part of the future answer on the input side}: the simulator generates what the supplied participants say, while their presence and number are already given.
\textit{Agreement with real semantic distributions then measures content generation conditional on participation, not prediction of participation itself.}
Stronger macro--micro coupling cannot remove this dependency.
Without participation prediction, the system remains a content renderer along a supplied trajectory rather than a joint forecaster of activity and communication.

\subsection*{Intervention-Aware Simulation: Conditioning vs. Endogenizing External Conditions}
Recent social simulators increasingly incorporate interventions into computational experiments. \textit{Social Simulacra} treats interventions as experimenter-specified conditions~\citep{park2022social}, whereas \textit{PolicySim} and \textit{IntervenSim} endogenize intervention selection within the simulated dynamics~\citep{policysim-WWW2026,zhang2026intervensim}. Such formulations support intervention optimization, but differ from conditional social forecasting when the external-event stream is specified: \textbf{internally selecting or modifying an intervention changes the prediction target}, replacing the declared condition with a generated scenario. Moreover, aggregate or group-level responses do not generate the underlying account-level participation and communication. \textit{Conditional social simulation therefore requires both adherence to specified external events and generation of the communications through which collective dynamics emerge}.

Across these strands, language is omitted from the forecast, participation is confined to an instantiated cast, future structure is supplied as input, or external conditions are generated inside an aggregated simulator.
The missing formulation follows a single boundary: \textbf{starting from observed history, condition on external events, predict participation and communication, and evaluate their induced observable trajectories}.

\section{Theoretical Foundations of the Simulation Target}
\label{app:B}
\label{app:motivation}

This section formalizes two requirements underlying the communication-centered formulation: the simulation target must be grounded in observable records, and participation must remain part of the predicted process.

\subsection{An Identification Limit for Unobserved Individual Dynamics}
\label{app:latent_identification}

Let $O$ collect the available platform records.
Consider two latent-process models $P_0$ and $P_1$, with observable laws $P_0^O$ and $P_1^O$.
Let $\tau_j=\tau(P_j)$ denote a scalar property of the hidden individual dynamics, such as the probability of a state change during inactivity, and write $\Delta=|\tau_1-\tau_0|$.

\begin{proposition}[Observation-limited identification]
\label{prop:latent_identification_bound}
For any estimator $\widehat{\tau}(O)$ based only on observable records,
\begin{equation}
\inf_{\widehat{\tau}}
\max_{j\in\{0,1\}}
\mathbb E_{P_j^O}
\left[
\bigl(\widehat{\tau}(O)-\tau_j\bigr)^2
\right]
\ge
\frac{\Delta^2}{4}
\left(
1-\operatorname{TV}(P_0^O,P_1^O)
\right),
\label{eq:latent_identification_bound}
\end{equation}
where $\operatorname{TV}$ denotes total variation distance.
In particular, observationally equivalent models with different latent dynamics imply an irreducible worst-case error of at least $\Delta^2/4$.
\end{proposition}

\begin{proof}
Let $p_0$ and $p_1$ be densities of the observable laws with respect to a common dominating measure $\mu$.
For an arbitrary estimator, define
\[
R_j
=
\mathbb E_{P_j^O}
\left[
\bigl(\widehat{\tau}(O)-\tau_j\bigr)^2
\right].
\]
Then
\begin{align}
R_0+R_1
&=
\int
\left[
p_0(o)(\widehat{\tau}(o)-\tau_0)^2
+
p_1(o)(\widehat{\tau}(o)-\tau_1)^2
\right]d\mu(o)
\nonumber\\
&\ge
\int
\min\{p_0(o),p_1(o)\}
\left[
(\widehat{\tau}(o)-\tau_0)^2
+
(\widehat{\tau}(o)-\tau_1)^2
\right]d\mu(o)
\nonumber\\
&\ge
\frac{\Delta^2}{2}
\int \min\{p_0(o),p_1(o)\}\,d\mu(o)
\nonumber\\
&=
\frac{\Delta^2}{2}
\left(
1-\operatorname{TV}(P_0^O,P_1^O)
\right).
\label{eq:latent_risk_derivation}
\end{align}
The second inequality follows from
\[
(a-\tau_0)^2+(a-\tau_1)^2
=
2\left(a-\frac{\tau_0+\tau_1}{2}\right)^2
+
\frac{\Delta^2}{2}.
\]
Using $\max\{R_0,R_1\}\ge(R_0+R_1)/2$ and taking the infimum over estimators proves Eq.~\eqref{eq:latent_identification_bound}.
\end{proof}

Sparse participation admits precisely this ambiguity.
Suppose an account is observed only at steps $0$ and $2$, with observations revealing its state at those steps.
The hidden paths $(0,0,0)$ and $(0,1,0)$ produce identical observable records but imply different values of
$\tau(P)=P(Z_1\neq Z_0)$.
Here $\Delta=1$ and $P_0^O=P_1^O$, giving a lower bound of $1/4$.

\textbf{Increasing simulator capacity cannot resolve a distinction absent from the observations.}
The bound concerns recovery of hidden individual dynamics, not prediction of subsequent communication.
This motivates defining the empirical target through observable events, while using participant-level latent representations as conditional computational features rather than identified psychological trajectories.

\subsection{The Prediction Burden Removed by Participation Conditioning}
\label{app:participation_prediction}

Fix a declared external stream $e_{1:T}$ and a common initial history.
Let $P^e$ and $Q^e$ denote the reference and simulator laws of the observable rollout.
At each step, write
\[
X_t=(\mathcal H_{t-1},e_t),
\qquad
S_t=\mathcal U_t,
\]
and let $C_t$ collect the participant-indexed actions and contents.
Thus $\mathcal Y_t$ is represented by $(S_t,C_t)$.
Both processes use the same deterministic update of the complete observable history.

Define the content-only discrepancy on reference histories as
\begin{equation}
\mathcal D_{\mathrm{comm}}
=
\sum_{t=1}^{T}
\mathbb E_{P^e}
\left[
\operatorname{KL}
\left(
p(C_t\mid S_t,X_t)
\,\|\,
q(C_t\mid S_t,X_t)
\right)
\right].
\label{eq:content_only_discrepancy}
\end{equation}

\begin{proposition}[Sequential participation--communication decomposition]
\label{prop:sequential_prediction_decomposition}
Assume $P^e\ll Q^e$ and finite divergences.
The full rollout discrepancy satisfies
\begin{align}
\operatorname{KL}(P^e\|Q^e)
&=
\mathcal D_{\mathrm{part}}
+
\mathcal D_{\mathrm{comm}},
\label{eq:rollout_discrepancy_decomposition}\\
\mathcal D_{\mathrm{part}}
&=
\sum_{t=1}^{T}
\mathbb E_{P^e}
\left[
\operatorname{KL}
\left(
p(S_t\mid X_t)
\,\|\,
q(S_t\mid X_t)
\right)
\right].
\label{eq:participation_discrepancy}
\end{align}
Consequently,
$\operatorname{KL}(P^e\|Q^e)\ge\mathcal D_{\mathrm{comm}}$,
with equality exactly when the participation kernels agree on reference histories almost surely.
\end{proposition}

\begin{proof}
Applying the probability chain rule over time and within each step gives
\begin{align}
\log
\frac{P^e(\mathcal Y_{1:T})}
     {Q^e(\mathcal Y_{1:T})}
&=
\sum_{t=1}^{T}
\log
\frac{p(S_t,C_t\mid X_t)}
     {q(S_t,C_t\mid X_t)}
\nonumber\\
&=
\sum_{t=1}^{T}
\left[
\log\frac{p(S_t\mid X_t)}{q(S_t\mid X_t)}
+
\log\frac{p(C_t\mid S_t,X_t)}
         {q(C_t\mid S_t,X_t)}
\right].
\end{align}
Taking expectation under $P^e$ and conditioning on $X_t$ for the first term and $(S_t,X_t)$ for the second yields
\begin{align}
\operatorname{KL}(P^e\|Q^e)
&=
\sum_{t=1}^{T}
\mathbb E_{P^e}
\left[
\operatorname{KL}
\left(
p(S_t\mid X_t)\,\|\,q(S_t\mid X_t)
\right)
\right]
\nonumber\\
&\quad+
\sum_{t=1}^{T}
\mathbb E_{P^e}
\left[
\operatorname{KL}
\left(
p(C_t\mid S_t,X_t)\,\|\,q(C_t\mid S_t,X_t)
\right)
\right].
\end{align}
Non-negativity of conditional KL gives the inequality.
Equality holds if and only if every participation divergence vanishes almost surely.
\end{proof}

\textbf{Conditioning content evaluation on realized participants leaves an entire component of process prediction untested.}
Even when $\mathcal D_{\mathrm{comm}}=0$, the full rollout discrepancy can remain positive through $\mathcal D_{\mathrm{part}}$.
The decomposition concerns the joint event-process law, rather than the KL between aggregate attitude distributions used in empirical evaluation.


\section{Derivations of the Simulation Method}
\label{app:C}

\subsection{Participation and the Observable Transition Kernel}
\label{app:participation_kernel}

Write $X_t=(\mathcal H_{t-1},\mathcal E_t)$ and
$M=|\mathcal U|$.
For a feasible predicted count $n=\hat N_t\in\{0,\ldots,M\}$, define positive activation weights
\[
w_i=q_\psi(u_i\mid p_i,X_t).
\]
Sequential sampling without replacement induces the ordered activation law
\begin{equation}
Q_\psi^{\mathrm{ord}}(u_{i_1},\ldots,u_{i_n}\mid n,X_t)
=
\prod_{j=1}^{n}
\frac{w_{i_j}}
{\sum_{\ell\notin\{i_1,\ldots,i_{j-1}\}}w_\ell}.
\label{eq:app_ordered_activation}
\end{equation}
The probability of an unordered participant set $S$ is
\begin{equation}
Q_\psi(S\mid n,X_t)
=
\sum_{\sigma\in\operatorname{Perm}(S)}
Q_\psi^{\mathrm{ord}}(\sigma\mid n,X_t).
\label{eq:app_set_activation}
\end{equation}

Given $S$, the state model produces
$(\hat m_t,\{h_{t,j}\})=f_\eta(X_t,P_t)$.
Let
\[
x_{t,j}
=
(X_t,p_{u_{t,j}},h_{t,j},\hat m_t)
\]
be the communication context, and let
$\bar\pi_{\phi,\gamma}$ denote the distribution induced by candidate generation and selection.
With independent candidate sets across activated participants, the complete step kernel is
\begin{equation}
Q_\theta(n,S,C_t\mid X_t)
=
Q_\omega(n\mid X_t)\,
Q_\psi(S\mid n,X_t)
\prod_{j=1}^{n}
\bar\pi_{\phi,\gamma}(a_{t,j},c_{t,j}\mid x_{t,j}),
\label{eq:app_step_kernel}
\end{equation}
where the point forecast corresponds to
$Q_\omega(n\mid X_t)=\mathbf 1\{n=f_\omega(X_t)\}$.

\begin{proposition}[Normalized open-population transition]
\label{prop:app_normalized_transition}
If $Q_\omega$ is supported on $\{0,\ldots,M\}$ and each selected communication law is normalized, Eq.~\eqref{eq:app_step_kernel} defines a normalized transition with exactly $n$ distinct active accounts.
\end{proposition}

\begin{proof}
At each activation draw, the conditional probabilities sum to one over the remaining accounts.
Successively summing Eq.~\eqref{eq:app_ordered_activation} over its final, penultimate, and preceding draws therefore gives one.
Grouping permutations by their underlying sets preserves normalization.
Consequently,
\begin{align}
\sum_{n=0}^{M}\sum_{S:|S|=n}\sum_{C_t}
Q_\theta(n,S,C_t\mid X_t)
&=
\sum_{n=0}^{M}Q_\omega(n\mid X_t)
\sum_{S:|S|=n}Q_\psi(S\mid n,X_t)
\nonumber\\
&=1.
\end{align}
For $n=0$, the participant set and communication collection are empty and the empty product equals one.
Exclusion operates within each step, allowing an account to reappear at later steps.
\end{proof}

For $n>0$, let $A(y)\in\{e_1,e_2,e_3\}$ encode the expressed attitude of communication $y$.
The realized state update is
\begin{equation}
m_t=\frac{1}{n}\sum_{j=1}^{n}A(y_{t,j}),
\qquad
r_t=g_\xi(r_{t-1},m_t,\mathcal Y_t,\mathcal E_t).
\label{eq:app_realized_state}
\end{equation}
Thus \textbf{$\hat m_t$ conditions generation, while $m_t$ records its realized outcome} and enters the next-step history.

\subsection{Differentiable Future-Constrained Learning}
\label{app:future_constrained_gradient}

For a training communication $y^\star$ and context $x$, let
$v_\phi(x,y^\star)$ collect the policy's teacher-forced prompt and response features.
The trajectory critic acts on these continuous features:
\begin{equation}
\widehat{\mathbf m}(v)
=
C_\gamma(x,v),
\qquad
d_\gamma(v)
=
\sum_{k:b_{t,k}=1}
D\!\left(m_{t+k}^\star,[\widehat{\mathbf m}(v)]_k\right),
\label{eq:app_masked_critic}
\end{equation}
where $b_{t,k}$ indicates an observed future target.
An empty valid-horizon set gives $d_\gamma(v)=0$.

Let $v_+$ be the matched feature vector and
$\{v_r^-\}_{r=1}^{R}$ its masked or shuffled counterparts.
Define
\begin{equation}
\begin{aligned}
B_t&=\mathbf 1\!\left\{\sum_k b_{t,k}>0\right\},\\
a_r&=\delta+d_\gamma(v_+)-d_\gamma(v_r^-),\\
\mathcal L_{\mathrm{sens}}
&=\frac{B_t}{R}\sum_{r=1}^{R}[a_r]_+.
\end{aligned}
\label{eq:app_sensitivity_objective}
\end{equation}
For a valid sample, identical matched and negative predictions incur loss $\delta$.
Zero sensitivity loss requires
\[
d_\gamma(v_r^-)\ge d_\gamma(v_+)+\delta
\qquad\text{for every }r.
\]

The joint objective in Eq.~\eqref{eq:policy_objective} becomes
\begin{equation}
\mathcal L_{\mathrm{joint}}
=
\mathcal L_{\mathrm{comm}}
+
\lambda_{\mathrm{traj}}d_\gamma(v_+)
+
\lambda_{\mathrm{sens}}
\frac{B_t}{R}\sum_{r=1}^{R}[a_r]_+.
\label{eq:app_joint_objective}
\end{equation}

\paragraph{Gradient transmission to the policy.}
Write
$J_+=\partial v_+/\partial\phi$,
$J_r^-=\partial v_r^-/\partial\phi$,
$g_+=\nabla_{v_+}d_\gamma(v_+)$, and
$g_r^-=\nabla_{v_r^-}d_\gamma(v_r^-)$.
Away from hinge points, the chain rule yields
\begin{align}
\nabla_\phi\mathcal L_{\mathrm{joint}}
={}&
\nabla_\phi\mathcal L_{\mathrm{comm}}
+
\lambda_{\mathrm{traj}}J_+^\top g_+
\nonumber\\
&+
\frac{\lambda_{\mathrm{sens}}B_t}{R}
\sum_{r=1}^{R}
\mathbf 1\{a_r>0\}
\left(
J_+^\top g_+-(J_r^-)^\top g_r^-
\right).
\label{eq:app_policy_gradient}
\end{align}
The critic gradient is
\begin{align}
\nabla_\gamma\mathcal L_{\mathrm{joint}}
={}&
\lambda_{\mathrm{traj}}\nabla_\gamma d_\gamma(v_+)
\nonumber\\
&+
\frac{\lambda_{\mathrm{sens}}B_t}{R}
\sum_{r=1}^{R}
\mathbf 1\{a_r>0\}
\left(
\nabla_\gamma d_\gamma(v_+)
-
\nabla_\gamma d_\gamma(v_r^-)
\right).
\label{eq:app_critic_gradient}
\end{align}
At $a_r=0$, a hinge subgradient is used.
Constant masked components have zero Jacobian; shuffled features retain their reindexed gradient paths.

Equations~\eqref{eq:app_policy_gradient}--\eqref{eq:app_critic_gradient} establish the coupling:
\textbf{future-trajectory supervision updates both the critic and the policy through continuous communication features}.
Observed tokens remain fixed during teacher forcing, while their hidden representations depend on the policy parameters.

\begin{table*}[t]
\centering
\small
\setlength{\tabcolsep}{5pt}
\renewcommand{\arraystretch}{1.15}
\caption{\textbf{Information coverage of social-simulation
benchmarks.}
Traffic: event-time aggregate activity;
Micro: individual communications;
Reversal: opinion-reversal events.
The final three columns indicate structured intervention
annotations at each level.
\dcmark: explicitly covered;
\dxmark: not specified as a benchmark component.}
\label{tab:dataset_comparison}
\begin{tabular*}{\textwidth}{
    @{\extracolsep{\fill}}lcccccc@{}
}
\toprule
& \multicolumn{3}{c}{\textbf{Observed social dynamics}}
& \multicolumn{3}{c}{\textbf{Intervention annotations}} \\
\cmidrule(lr){2-4}
\cmidrule(l){5-7}
\textbf{Benchmark}
& \textbf{Traffic}
& \textbf{Micro}
& \textbf{Reversal}
& \textbf{Event}
& \textbf{Platform}
& \textbf{Participant} \\
\midrule
GA-S$^3$~\citep{GAS32025-ga}
& \dcmark & \dxmark & \dxmark
& \dxmark & \dxmark & \dxmark \\

MF-LLM~\citep{MF-LLM2025mf}
& \dxmark & \dcmark & \dxmark
& \dxmark & \dxmark & \dxmark \\

M$^2$Sim~\citep{zhang2026coupling}
& \dxmark & \dcmark & \dcmark
& \dxmark & \dxmark & \dxmark \\
\midrule
\textbf{This work}
& \dcmark & \dcmark & \dcmark
& \dcmark & \dcmark & \dcmark \\
\bottomrule
\end{tabular*}
\end{table*}

\begin{figure*}[t]
    \centering
    \includegraphics[width=\textwidth]{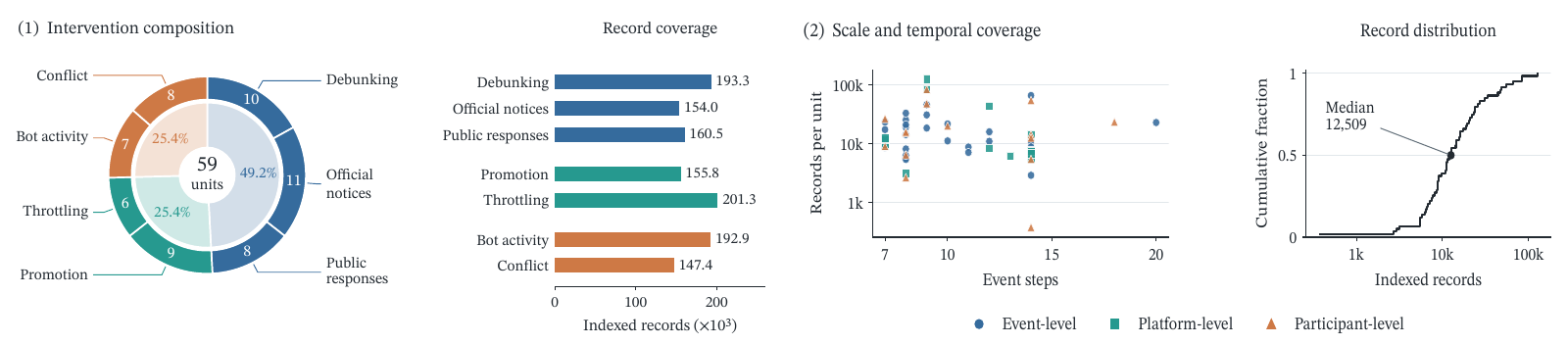}
    \caption{\textbf{Dataset composition and scale.}
    \textbf{(1)} Inner and outer rings show unit counts across
    three intervention levels and seven types, respectively.
    Matching colors identify types in the adjacent bars,
    which report indexed records in thousands.
    Legend percentages denote unit shares by intervention level.
    \textbf{(2)} Each point shows one unit's duration and
    indexed-record count; the empirical cumulative distribution
    summarizes record counts across units (median: 12,509).
    Record-count axes in (2) are logarithmic.}
    \label{fig:dataset_overview}
\end{figure*}

\subsection{Future-Blind Lookahead and Candidate Selection}
\label{app:future_blind_selection}

Let $\mathcal I_t$ contain the frozen models, admissible participant information, simulated history, and environmental inputs declared available at step $t$.
The state-model lookahead produces a reference trajectory
\begin{equation}
\widetilde{\mathbf m}_{t+1:t+K}
=
R_\eta(\mathcal I_t,S_t,\hat m_t;\zeta_t),
\label{eq:app_predicted_reference}
\end{equation}
where $\zeta_t$ is an inference seed.
All quantities used in the lookahead are declared inputs or internally generated predictions.

For a nonempty candidate set
$\mathcal C_{t,j}=\{y^{(1)},\ldots,y^{(J)}\}$,
selection uses
\begin{equation}
\begin{aligned}
s_t(y)
&=
D_K\!\left(
C_\gamma(x_{t,j},v_\phi(x_{t,j},y)),
\widetilde{\mathbf m}_{t+1:t+K}
\right),\\
\widehat y_{t,j}
&\in\arg\min_{y\in\mathcal C_{t,j}}s_t(y),
\end{aligned}
\label{eq:app_candidate_selection}
\end{equation}
with a fixed tie-breaking rule.
Here $D_K$ aggregates trajectory discrepancies across predicted horizons.

\begin{proposition}[Future-blind rollout]
\label{prop:app_future_blind_rollout}
If every inference module depends only on $\mathcal I_t$ and inference randomness, two runs with identical admissible inputs, frozen parameters, and random seeds produce identical rollouts regardless of the withheld future records.
\end{proposition}

\begin{proof}
At the first step, both runs produce identical counts, participant sets, state predictions, candidate sets, and lookahead references.
Equation~\eqref{eq:app_candidate_selection} therefore selects identical communications, yielding identical updated histories.
Induction gives equality at every subsequent step.
\end{proof}

\begin{proposition}[Finite-candidate selection bound]
\label{prop:app_selection_bound}
Let $s_t^\star$ be a target compatibility score on $\mathcal C_{t,j}$.
If
$\sup_{y\in\mathcal C_{t,j}}|s_t(y)-s_t^\star(y)|\le\epsilon$,
then
\begin{equation}
s_t^\star(\widehat y_{t,j})
-
\min_{y\in\mathcal C_{t,j}}s_t^\star(y)
\le 2\epsilon.
\label{eq:app_selection_regret}
\end{equation}
\end{proposition}

\begin{proof}
For $y^\star\in\arg\min_{y\in\mathcal C_{t,j}}s_t^\star(y)$,
\[
s_t^\star(\widehat y_{t,j})
\le s_t(\widehat y_{t,j})+\epsilon
\le s_t(y^\star)+\epsilon
\le s_t^\star(y^\star)+2\epsilon.
\]
\end{proof}

The first result establishes the inference information boundary.
The second quantifies selection accuracy within the generated candidate set under bounded scoring error.
Together, they separate \textbf{future-supervised learning} from \textbf{future-blind generation and selection}.

\begin{table*}[t]
\centering
\small
\setlength{\tabcolsep}{4pt}
\renewcommand{\arraystretch}{1.15}
\caption{\textbf{Simulation settings and participant pools.}
Warm-up is measured in interactions and content length in tokens.
Replacement refers to sampling within a simulation step.}
\label{tab:simulation_protocol}

\begin{tabular*}{\textwidth}{
    @{\extracolsep{\fill}}ccccc@{}
}
\toprule
\multicolumn{5}{@{}l}{
    \textbf{Generation and candidate selection}
} \\
\midrule
\textbf{Warm-up}
& \textbf{Seed}
& \textbf{Candidates $J$}
& \textbf{Temperature}
& \textbf{Top-$p$} \\
50 & 46 & 4 & 0.8 & 0.9 \\
\addlinespace[4pt]

\shortstack{\textbf{Repetition}\\\textbf{penalty}}
& \shortstack{\textbf{Content}\\\textbf{limit}}
& \shortstack{\textbf{MC-dropout}\\\textbf{samples}}
& \shortstack{\textbf{Uncertainty}\\\textbf{weight}}
& \shortstack{\textbf{Quality}\\\textbf{weight}} \\
1.05 & 48 & 4 & 0.10 & 0.05 \\
\midrule

\multicolumn{5}{@{}l}{
    \textbf{Participant-pool comparison}
} \\
\midrule
\textbf{Participant pool}
& \textbf{Size}
& \textbf{Sampling unit}
& \textbf{Sampling}
& \textbf{Replacement} \\
Event-specific
& Event-dependent
& User account
& Uniform
& No \\
External (SocioVerse)
& 15,992
& User account
& Uniform
& No \\
\bottomrule
\end{tabular*}
\end{table*}

\section{Detailed Experimental Setup}
\label{app:D}

\subsection{Dataset and Preprocessing}
\label{app:dataset}

The dataset combines \textbf{traffic dynamics, individual
communications, opinion reversals, and timed interventions}
across \textbf{59 intervention-event units}, comprising
1,205,150 indexed interaction records and 606 discrete event
steps.
Each unit links communication texts and participant features
to interaction counts, attitude trajectories, and intervention
timing; different intervention contexts of the same underlying
topic form separate units.
Table~\ref{tab:dataset_comparison} compares benchmark coverage,
and Figure~\ref{fig:dataset_overview} summarizes the dataset's
composition and scale.

\noindent\textbf{Intervention taxonomy.}
Seven intervention types are organized by
\textbf{where they enter the social process}.
\textbf{Event-level} interventions update event information:
\emph{debunking} corrects circulating claims,
\emph{official notices} provide institutional statements,
and \emph{public responses} capture explanations from
involved parties.
\textbf{Platform-level} interventions regulate exposure
through \emph{promotion} or \emph{throttling}.
\textbf{Participant-level} interventions comprise
\emph{bot activity}, introducing automated participation,
and \emph{conflict}, capturing antagonistic interactions.
Their occurrence steps are aligned with the communication
and attitude trajectories.

\noindent\textbf{Preprocessing and supervision.}
Training, validation, and test contain
\textbf{44, 8, and 7 intervention-event units}, respectively.
Training and validation yield 27,757 and 5,270 policy examples;
the test set spans five underlying topics, 283,952 indexed
records, and 68 event steps.
Each example conditions on available event context,
participant features, the preceding mean-field synopsis,
and an activation-local representation.
The target consists of a \texttt{comment} or \texttt{share}
action and its text.
Attitude distributions use the order
$[\mathrm{positive},\mathrm{neutral},\mathrm{negative}]$.
Future supervision comprises cumulative attitude distributions
over the next $K=10$ steps, with unavailable horizons masked.


\subsection{Simulation Protocol}
\label{app:simulation_protocol}

\noindent\textbf{Initialization and rollout.}
Each event retains its post context and begins with a
50-interaction warm-up, followed by participant sampling
and policy generation.
Selected communications update the simulated history
and mean field.
Table~\ref{tab:simulation_protocol} summarizes the settings.

\noindent\textbf{Participant pools.}
OpenSim uses an \emph{event-specific participant pool}
constructed from the corresponding event's observed participants,
whereas OpenSim$^\dagger$ uses a frozen \emph{external participant pool}
of 15,992 real-user profiles from
SocioVerse~\citep{Socioverse2025socioverse}.
Both settings uniformly sample user accounts without replacement
within each step, while allowing users to participate again
in subsequent steps.
Event context and stepwise communication counts are held fixed
in the pool comparison.
Neither setting consumes the observed future participant sequence.

\noindent\textbf{Communication generation.}
At each policy call, four candidate communications are generated.
Selection minimizes masked KL divergence to the state model's
predicted future attitude trajectory, together with
MC-dropout uncertainty and output-quality penalties.

\begin{table*}[t]
\centering
\footnotesize
\setlength{\tabcolsep}{3pt}
\renewcommand{\arraystretch}{1.12}
\caption{\textbf{Key model and training hyperparameters.}
Hyperparameters are shared across backbone configurations unless otherwise specified.
Epoch counts describe the initial training schedule before validation-based checkpoint selection.}
\label{tab:training_hparams}

\begin{tabular*}{\textwidth}{
    @{\extracolsep{\fill}}lclclc@{}
}
\toprule
\multicolumn{2}{c}{\textbf{State model}}
& \multicolumn{2}{c}{\textbf{Communication policy}}
& \multicolumn{2}{c}{\textbf{Trajectory critic}} \\
\cmidrule(lr){1-2}
\cmidrule(lr){3-4}
\cmidrule(lr){5-6}

Hidden dimension & 192
& LoRA rank / alpha & 32 / 32
& Hidden dimension & 256 \\

Layers / heads & 2 / 6
& Sequence length & 704
& Recurrent layers & 1 \\

Latent dimension & 64
& Effective batch size & 32
& Dropout & 0.2 \\

Participants / step & 128
& SFT learning rate & $8\times10^{-6}$
& Future horizon $K$ & 10 \\

Learning rate & $10^{-3}$
& Joint learning rate & $3\times10^{-6}$
& MC-dropout samples & 4 \\

Epoch budget & 12
& SFT epochs & 3
& Warm-up learning rate & $5\times10^{-5}$ \\

Auxiliary horizon & 3
& Initial joint epochs & 2
& Joint learning rate & $3\times10^{-5}$\\

\midrule
\multicolumn{6}{@{}l}{
    \textbf{Loss coefficients and weighting schedule}
} \\
\midrule

Increment weight & 0.35
& Direction weight & 0.15
& Curvature weight & 0.05 \\

MC-consistency weight & 0.01
& Relative sensitivity weight & 0.5
& Sensitivity margin $\delta$ & 0.05 \\

Maximum future weight & 0.15
& Horizon discount & 0.95
& Future-weight ramp & 20\% \\

\bottomrule
\end{tabular*}
\end{table*}

\begin{figure*}[t]
\centering
\includegraphics[width=\linewidth]{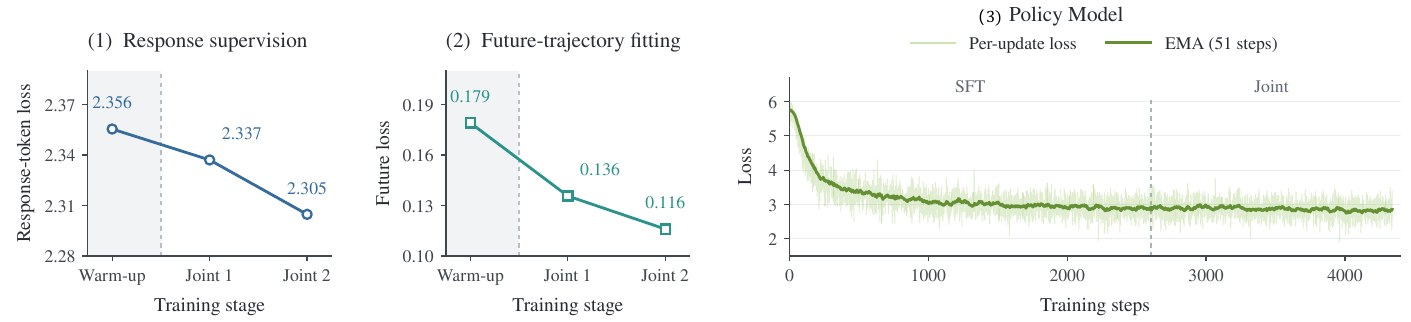}
\caption{\textbf{Training dynamics of Qwen2.5-1.5B.}
\textbf{(1) Response:} epoch-averaged response-token loss.
\textbf{(2) Future:} epoch-averaged future-fitting objective
$\mathcal{L}_{\mathrm{fit}}$.
Both panels cover critic warm-up and two joint-training epochs.
\textbf{(3) Policy:} stage-specific training loss over
2,604 SFT updates and 1,736 joint-training updates.
The light curve shows the mean mini-batch loss per update;
the dark curve shows an exponential moving average
with a span of 51 updates.
The dashed line marks the SFT--joint boundary,
where smoothing restarts.}
\label{fig:training_curves}
\end{figure*}

\subsection{Training Curves and Hyperparameters}
\label{app:training}

\noindent\textbf{Training configuration.}
Communication policies are adapted from Qwen2-1.5B,
Qwen2.5-1.5B/3B/7B, and Qwen3-1.7B using LoRA\citep{li2026lora}.
All backbones follow a three-stage framework:
communication SFT, critic-only warm-up, and joint optimization.
Warm-up freezes the policy and trains the trajectory critic;
joint training updates both the policy adapters and critic.
Checkpoint selection uses validation communication quality
and response-sensitive trajectory prediction.
Additional critic calibration, when required by validation,
keeps the policy frozen.

\noindent\textbf{Optimization.}
Table~\ref{tab:training_hparams} reports the key model
hyperparameters, training schedules, and loss coefficients.
All modules are optimized with AdamW using random seed 42.
The future-objective weight increases linearly during
the initial portion of joint training, while unavailable
future horizons are masked throughout trajectory supervision.

\noindent\textbf{Training dynamics.}
Figure~\ref{fig:training_curves} separates response supervision
and future fitting from the overall policy optimization trace,
showing component losses at the epoch level and policy loss
at the optimizer-update level.

\begin{table*}[t]
\centering
\footnotesize
\begingroup
\definecolor{FullTint}{HTML}{EDE4F7}
\definecolor{AblationTint}{HTML}{F5F5F5}

\caption[Component ablations across backbones.]{%
\textbf{Component ablations across backbones.}
\textsc{w/o ST Objective} removes state-transition supervision, while
\textsc{w/o Future Constraint} removes the future-constrained policy objective.
\textbf{Bold}/\underline{underline} indicate the best/second-best result within each backbone.}
\label{tab:component_ablation}

\scriptsize
\setlength{\tabcolsep}{3.5pt}
\renewcommand{\arraystretch}{0.98}
\setlength{\aboverulesep}{1pt}
\setlength{\belowrulesep}{1pt}

\begin{tabularx}{\textwidth}{
@{}ll*{6}{>{\centering\arraybackslash}X}@{}}
\toprule
\textbf{Model} &
\textbf{Setting} &
\textbf{KL} $\downarrow$ &
\textbf{Wass.} $\downarrow$ &
\textbf{DTW} $\downarrow$ &
\mbox{\textbf{Macro-F1} $\uparrow$} &
\mbox{\textbf{Micro-F1} $\uparrow$} &
\textbf{NLL} $\downarrow$ \\
\midrule

\rowcolor{FullTint}
\cellcolor{white} &
\textbf{OpenSim}
& \textbf{1.1352} & \textbf{0.1922} & \textbf{0.1261}
& \textbf{0.4444} & \textbf{0.6616} & \textbf{4.3377} \\

\rowcolor{AblationTint}
\cellcolor{white} &
\textsc{w/o ST Objective}
& 3.1716 & 0.3074 & 0.2644
& 0.4087 & 0.5079 & 4.4791 \\

\rowcolor{AblationTint}
\cellcolor{white}
\multirow{-3}{*}{\shortstack{Qwen\\2-1.5B}}
& \textsc{w/o Future Constraint}
& \underline{1.6934} & \underline{0.2410} & \underline{0.1914}
& \underline{0.4331} & \underline{0.5597} & \underline{4.4315} \\
\midrule

\rowcolor{FullTint}
\cellcolor{white} &
\textbf{OpenSim}
& \textbf{1.0667} & \textbf{0.1724} & \textbf{0.1252}
& \textbf{0.6097} & \textbf{0.6252} & 4.5980 \\

\rowcolor{AblationTint}
\cellcolor{white} &
\textsc{w/o ST Objective}
& 2.9886 & 0.2725 & 0.2494
& 0.4222 & 0.5299 & \underline{4.4942} \\

\rowcolor{AblationTint}
\cellcolor{white}
\multirow{-3}{*}{\shortstack{Qwen\\2.5-1.5B}}
& \textsc{w/o Future Constraint}
& \underline{1.5839} & \underline{0.1940} & \underline{0.1864}
& \underline{0.5369} & \underline{0.5646} & \textbf{4.3211} \\
\midrule

\rowcolor{FullTint}
\cellcolor{white} &
\textbf{OpenSim}
& \textbf{1.0929} & \textbf{0.1780} & \textbf{0.1098}
& \textbf{0.5445} & \textbf{0.7610} & \underline{4.4414} \\

\rowcolor{AblationTint}
\cellcolor{white} &
\textsc{w/o ST Objective}
& 2.8526 & 0.2775 & 0.2412
& \underline{0.4782} & 0.5776 & 4.5111 \\

\rowcolor{AblationTint}
\cellcolor{white}
\multirow{-3}{*}{\shortstack{Qwen\\3-1.7B}}
& \textsc{w/o Future Constraint}
& \underline{1.7857} & \underline{0.2320} & \underline{0.1919}
& 0.4413 & \underline{0.6587} & \textbf{4.0974} \\
\bottomrule
\end{tabularx}
\endgroup
\end{table*}

\section{Additional Experimental Results and Analysis}
\label{app:E}

\begin{figure*}[t]
    \centering
    \includegraphics[width=\textwidth]{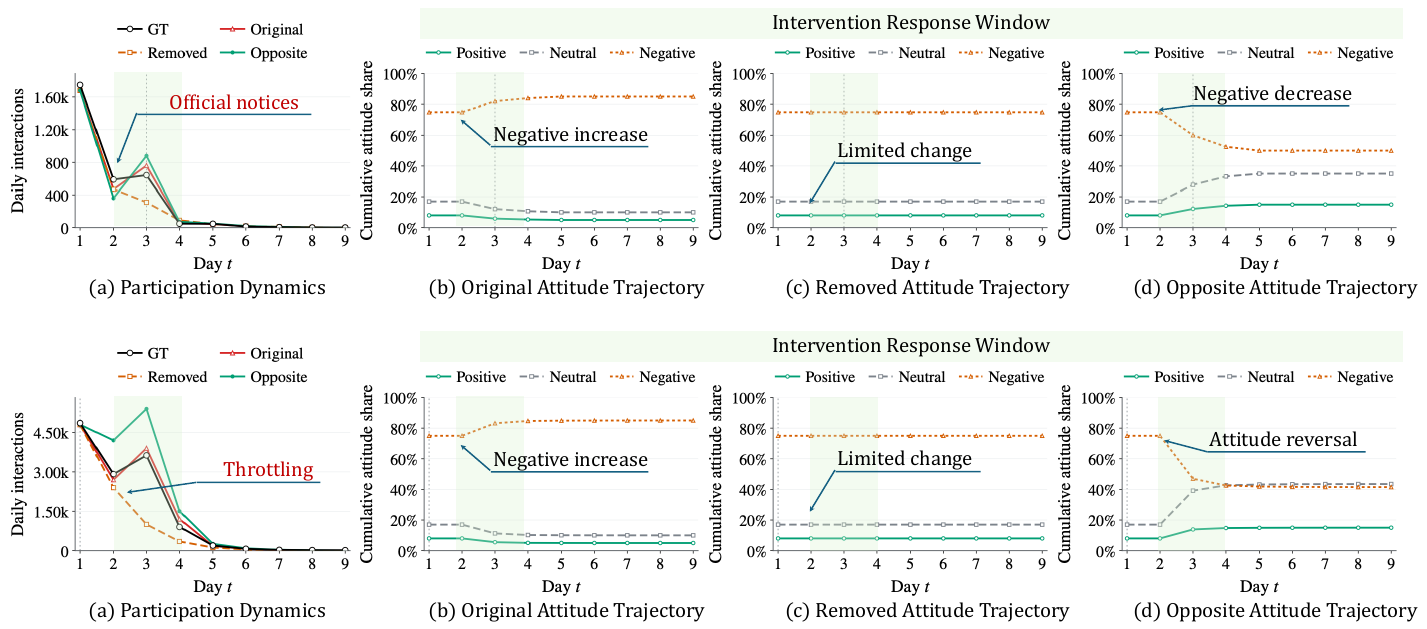}
    \caption{\textbf{Intervention responses in the Xi'an metro case.}
    The top and bottom rows examine official notices and
    throttling, respectively.
    Each row shows participation dynamics and cumulative
    attitude shares under \textsc{Original}, \textsc{Removed},
    and \textsc{Opposite}.
    GT denotes observed participation; shaded regions mark
    intervention response windows.}
    \label{fig:case_xian}
\end{figure*}

\begin{figure*}[t]
    \centering
    \includegraphics[width=\textwidth]{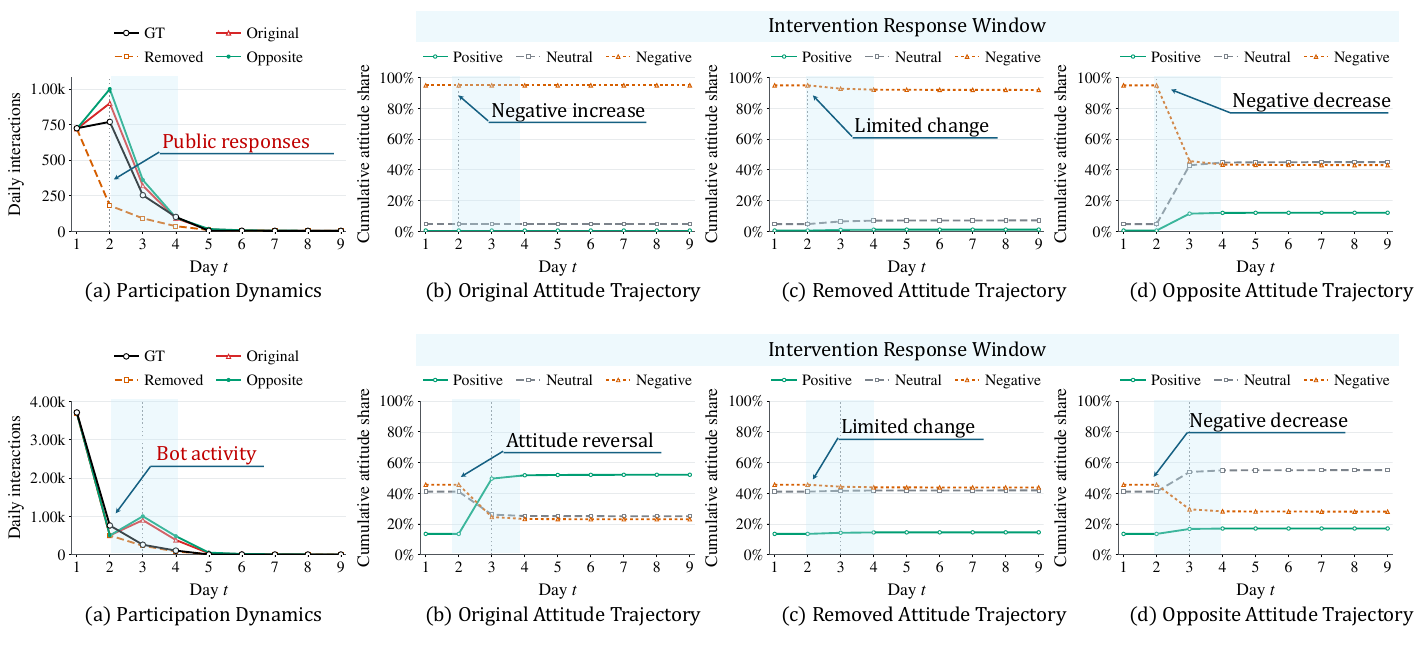}
    \caption{\textbf{Intervention responses in the Guilin tourism case.}
    The top and bottom rows examine public responses and
    bot activity, respectively.
    Each row shows participation dynamics and cumulative
    attitude shares under \textsc{Original}, \textsc{Removed},
    and \textsc{Opposite}.
    GT denotes observed participation; shaded regions mark
    intervention response windows.}
    \label{fig:case_guilin}
\end{figure*}

\subsection{Case Studies of Intervention Responses}
\label{app:intervention_cases}

Figures~\ref{fig:case_xian} and~\ref{fig:case_guilin}
compare participation and cumulative attitude trajectories
under \textsc{Original}, \textsc{Removed}, and \textsc{Opposite}.
The original conditions assess reproduction of observed
attitude trends, removal examines evolution without the
corresponding intervention, and opposite conditions explore
alternative governance responses.

In the Xi'an metro case, the initial response emphasized
the passenger's disruption of public order, while public
criticism focused on the security guard's conduct.
To address this mismatch, the opposite official notice
explicitly acknowledges improper handling, apologizes
to the passenger, and announces investigation,
accountability, and corrective measures.
For throttling, the opposite condition removes additional
exposure restrictions and allows discussion to spread naturally.
These changes test whether addressing public concerns
and maintaining open discussion can alleviate negative attitudes.
The original conditions reproduce the observed persistence
and further increase of negative attitudes.
Removing the interventions weakens the participation response
and leaves attitude shares comparatively stable.
Under the opposite conditions, the negative share declines
and the neutral share increases, with neutrality becoming
slightly dominant when throttling is lifted.
Increased participation therefore need not coincide
with more negative attitudes.

In the Guilin tourism case, a tourist's complaint about
dirty concert seats was followed by blocking by the official
tourism account, extending the controversy from service
quality to complaint handling.
The subsequent official response did not resolve public
dissatisfaction.
The opposite response therefore emphasizes concrete remedies:
unblocking and contacting the complainant, reporting
cleaning inspections and complaint-resolution progress,
and providing a follow-up feedback channel.
This tests whether verifiable corrective measures can
alleviate persistent dissatisfaction.
The original response condition reproduces the observed
dominance of negative attitudes.
Removing the response produces little change, whereas
the opposite response substantially reduces the negative
share, with neutrality becoming the largest category.
For bot activity, we replace positive stance promotion
with neutral factual communication, restricting messages
to verified information and updates on corrective actions.
This comparison examines how the two communication
strategies shape attitude evolution.
The original condition reproduces the observed shift
from negative to positive dominance, while the opposite
condition yields a neutral-dominated trajectory.
Together, these cases demonstrate reproduction of observed
attitude trends and simulation of possible outcomes
under alternative governance responses.

\begin{figure*}[t]
    \centering
    \includegraphics[width=\textwidth]{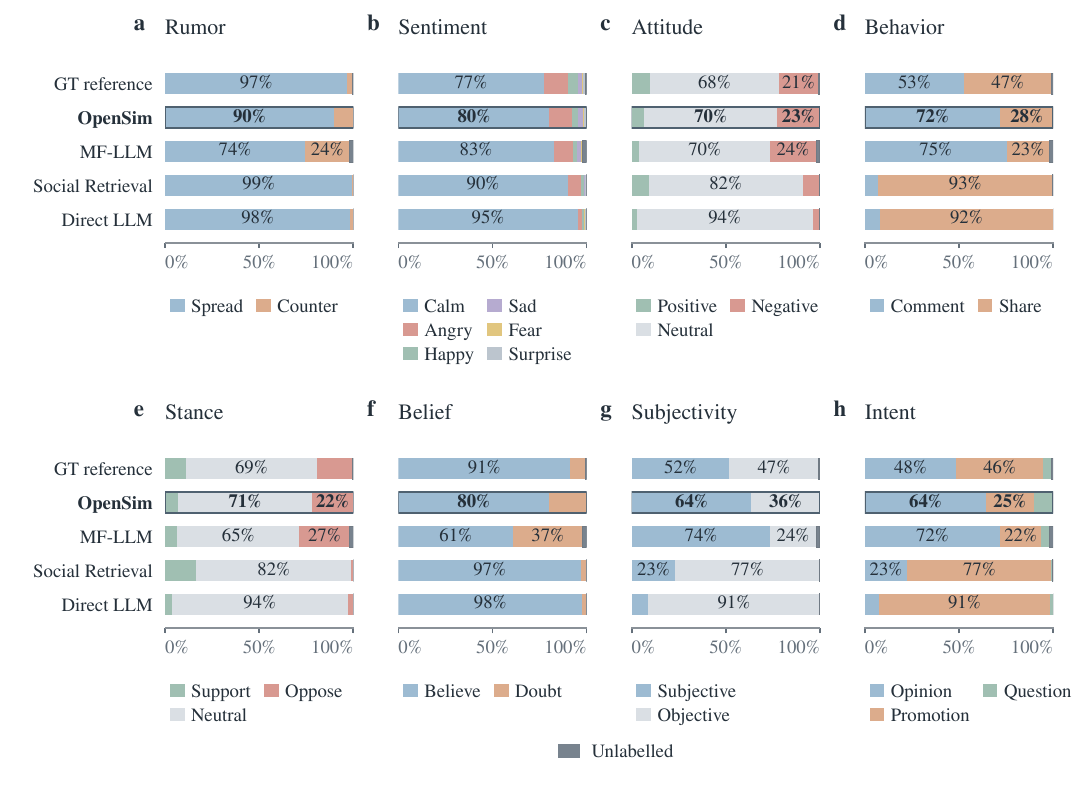}
    \caption{\textbf{Full semantic distributions across eight dimensions.}
    Qwen2-1.5B results pooled over seven test units and
    68 event-day groups, with 10,879 record pairs per method.
    Each bar shows record-weighted category proportions;
    gray denotes unlabelled records.
    GT reference averages the real-label distributions
    across the four evaluation caches.
    Percentage labels are shown for segments of at least 20\%.}
    \label{fig:semantic_distributions}
\end{figure*}

\subsection{Fine-Grained Semantic Distributions}
\label{app:semantic_distributions}

Figure~\ref{fig:semantic_distributions} compares complete
communication distributions across eight semantic dimensions.
\textbf{The clearest differences concern how agents communicate,
what they express, and how strongly their outputs concentrate
on dominant categories.}

\noindent\textbf{Behavior and communicative intent.}
Direct LLM and Social Retrieval allocate approximately
92\% of their outputs to sharing, compared with 46.8\%
in the reference distribution.
Their intent distributions are similarly concentrated
on promotion, at 91.1\% and 76.6\%, respectively,
against a reference share of 46.4\%.
OpenSim produces substantially more comments and opinions,
bringing both distributions closer to the reference.
The discrepancy therefore extends beyond action frequency:
it also changes the semantic composition of the simulated
communications.

\noindent\textbf{Expression style and affective composition.}
OpenSim retains 63.7\% subjective expression, compared with
52.0\% in the reference, whereas Direct LLM and Social Retrieval
produce predominantly objective outputs.
A related concentration appears in emotion and attitude~\citep{zheng2026surfacing}.
The calm category accounts for 95.2\% of Direct LLM outputs
and 90.0\% of Social Retrieval outputs, but 77.1\% of the
reference; OpenSim yields 80.3\%.
Its neutral-attitude share of 70.4\% likewise remains close
to the reference value of 68.2\%.
The stance distribution exhibits the same pattern,
with OpenSim retaining substantial opposing responses
alongside neutral communication.
\textbf{These comparisons reveal differences that matching
the dominant category alone would obscure.}

\noindent\textbf{Rumor and belief.}
OpenSim assigns more probability mass to countering and doubt
than Direct LLM and Social Retrieval.
These categories also exceed their reference proportions,
showing that broader response coverage does not automatically
imply closer distributional alignment.
The complete profiles distinguish this redistribution from
the stronger reference agreement observed in behavior,
expression style, and affective composition.

\end{document}